\documentclass[10pt,draftcls,onecolumn]{IEEEtran}

\usepackage{cite}

\usepackage{colortbl}
\usepackage{color}
\usepackage{xcolor}
\usepackage{bm}
\usepackage{bbm}
\usepackage{amsmath}
\usepackage{multirow}

\usepackage{cases}

\usepackage{amssymb}
\usepackage{amsthm}
\usepackage{graphicx}
\usepackage{mathrsfs}
\usepackage{empheq}	
\usepackage[mathcal]{euscript}
\usepackage[margin = 2cm]{geometry}
\usepackage{framed}
\usepackage{comment}

\usepackage{dsfont}
\usepackage{xr}
\makeatletter
\newcommand*{\addFileDependency}[1]{
  \typeout{(#1)}
  \@addtofilelist{#1}
  \IfFileExists{#1}{}{\typeout{No file #1.}}
}
\makeatother

\DeclareFontFamily{U}{mathx}{\hyphenchar\font45}
\DeclareFontShape{U}{mathx}{m}{n}{
      <5> <6> <7> <8> <9> <10>
      <10.95> <12> <14.4> <17.28> <20.74> <24.88>
      mathx10
      }{}
\DeclareSymbolFont{mathx}{U}{mathx}{m}{n}
\DeclareFontSubstitution{U}{mathx}{m}{n}
\DeclareMathAccent{\widecheck}{0}{mathx}{"71}

\newtheorem{theorem}{Theorem}
\newtheorem{lemma}{Lemma}
\newtheorem{corollary}{Corollary}

\newtheorem{assumption}{Assumption}

\newtheorem{definition}{Definition}
\newtheorem{example}{Example}

\newcommand{\beq}{\begin{equation}}
\newcommand{\eeq}{\end{equation}}
\newcommand{\beqa}{\begin{IEEEeqnarray}{rCl}}
\newcommand{\eeqa}{\end{IEEEeqnarray}}

\newcommand{\thetatrue}{\theta^{\bullet}}

\title{Performance Evaluation of Social Learning
\thanks{
F. Scala, M. Carpentiero, and V. Matta are with
DIEM, University of Salerno, I-84084 Fisciano (SA), Italy, and also with CNIT, Italy.

A. H. Sayed is with
EPFL, CH-1015 Lausanne, Switzerland.

The work of M. Carpentiero and V. Matta was partially supported by the European Union under the Italian National Recovery and Resilience Plan (NRRP) of NextGenerationEU, partnership on ``Telecommunications of the Future'' (PE00000001 - program ``RESTART'').

A short conference version of this article was presented at EUSIPCO 2025~\cite{ourEUSIPCO2025}.
}
}

\author{\IEEEauthorblockN{Felice Scala, Marco Carpentiero, Vincenzo Matta, and Ali H. Sayed}
}

\begin{document}
\maketitle

\begin{abstract}
Social Learning is a decentralized decision-making paradigm in which spatially dispersed agents collect streaming observations regulated by one of a finite number of models ({\em the hypotheses}).
The agents are interested in assigning probability scores (\emph{the beliefs}) to the possible hypotheses.
To this end, the agents exchange their beliefs according to a certain communication graph.
It has been shown that, under reasonable conditions on the identifiability of the decision model and the network connectivity, each agent ultimately places all the belief mass on the true hypothesis governing the data. 
However, several questions remain unanswered regarding the evaluation of the social learning performance.
One recently adopted performance metric is the rejection rate, i.e., the rate at which the beliefs about the erroneous hypotheses vanish. One contribution of this work is to establish that the rejection rate leads to several paradoxes, which make it unsuitable as a valid performance measure. 
We then focus on studying the error probability measure. 
For a binary Gaussian problem, we derive an analytical formula characterizing the ratio between the individual agents' probabilities and the optimal Bayesian probability. The formula shows that this ratio is expressed by the product of two terms quantifying the effect of the network connectivity and the role of the prior information. As a result, an \emph{irreducible gap} emerges between the decentralized and the centralized error probabilities, which is \emph{agent-dependent and does not disappear asymptotically}. 
\end{abstract}


\section{Introduction and Related Work}

\IEEEPARstart{I}{n} recent years, the theory of \emph{social learning (SL)} has been developed to design \emph{decentralized decision-making} systems \cite{mattaSocialLearningOpinion2024, bordignonSociallyIntelligentNetworks2024,zhaoLearningSocialNetworks2012,jadbabaieNonBayesianSocialLearning2012,nedicFastConvergenceRates2017,lalithaSocialLearningDistributed2018,Jadbabaie2018,Djuric2015,krishnamurthySocialLearningBayesian2013,chamleyRationalHerdsEconomic2004}. 
These are multi-agent systems where the agents are connected and allowed to share information according to a certain network graph. 
The agents collect from the environment streaming observations about some phenomenon of interest. The collected data are random, and their statistical behavior is governed by one from among a finite collection of models (called \emph{the hypotheses}). The true model is unknown to the agents.
The goal of each agent in social learning is to assign a probability score (called \emph{the belief}) to any hypothesis within the admissible set. 
Then, each agent decides in favor of the hypothesis that maximizes its belief. 
For several reasons, which also depend on the specific application (e.g., privacy requirements in a social network, communication/power/memory constraints in a sensor network), the agents are often prevented from sharing the raw data or their private decision models. They are allowed to share their beliefs, in particular, locally with their neighbors according to the network graph.

Several useful results are available regarding the theory of social learning --- see \cite{mattaSocialLearningOpinion2024, bordignonSociallyIntelligentNetworks2024} for overviews. 
The earlier implementations focused on arithmetic averaging schemes based on the diffusion \cite{zhaoLearningSocialNetworks2012} or consensus \cite{jadbabaieNonBayesianSocialLearning2012} paradigms, to pool the beliefs of neighboring agents. Subsequent studies considered the geometric-averaging pooling rule \cite{nedicFastConvergenceRates2017,lalithaSocialLearningDistributed2018, Jadbabaie2018}, which we will refer to as \emph{traditional SL} in the following.
 
The basic result in SL theory is the following consistency result that applies to different SL implementations: over a connected network (loosely meaning that any two agents are connected by a certain path over the graph) and assuming that the decision models are globally identifiable (loosely meaning that for each hypothesis there exists at least one agent whose data are informative to distinguish that hypothesis), each agent ultimately places all the belief mass on the true hypothesis. 
However, consistency only reveals that for sufficiently large observation windows each agent learns well. Several open points remain regarding \emph{how well the agents learn}, including: the choice of the performance metric, the analytical characterization of the performance, and the comparison among different strategies. 

For example, in traditional SL it has been shown that the beliefs about the wrong hypotheses vanish exponentially fast, with a certain \emph{rejection rate (RR)} \cite{nedicFastConvergenceRates2017,lalithaSocialLearningDistributed2018}. 
Starting from this fact, the rejection rate has been adopted in some works as a performance proxy to compare different strategies \cite{MertAAGA, mitraNewApproachDistributed2021}. Specifically, a strategy proposed in \cite{mitraNewApproachDistributed2021} (that we refer to as the \emph{minimum-belief strategy}) is shown to exhibit a larger rejection rate than traditional SL. Based on this result, superiority of the minimum-belief strategy over traditional SL was claimed. 
One contribution of this work is to show that the rejection rate is not a valid performance metric, since it leads to several paradoxes. In particular, we will show that any rejection rate is achievable, and that, if the rejection rate were used, a clairvoyant system implementing the true posterior would be outperformed by (infinitely many) other constructions where the beliefs can be arbitrarily distant from the true posterior. 

Therefore, we turn to examine SL strategies with a well-established criterion borrowed from decision theory, namely, the \emph{error probability}. 
Previous studies~\cite{ nedicFastConvergenceRates2017,lalithaSocialLearningDistributed2018,mattaSocialLearningOpinion2024, bordignonSociallyIntelligentNetworks2024,HuangWang,bordignonSocialLearningNonBayesian2023} 
have characterized the exponential rate at which the error probability \emph{(not the belief)} vanishes with time. 
When the observations are independent over time and across the agents, traditional social learning and the recent non-Bayesian-non-Bayesian (NB$^2$) strategy are able to attain the same error exponent as the optimal centralized decision strategy~\cite{mattaSocialLearningOpinion2024, bordignonSociallyIntelligentNetworks2024,HuangWang,bordignonSocialLearningNonBayesian2023}.

In this work, we revisit more carefully the impact of this error-exponent optimality, showing that it hides important factors. 
The relevant problem of examining the higher-order effects on the error probability has been recently addressed in~\cite{HuangWang}. 
For the case of a binary detection problem, the Authors compute closed-form upper bounds on the asymptotic ratio between the individual agents' probabilities and the optimal error probability. It is shown in~\cite{HuangWang} that the obtained bounds are not tight. Thus, the results from~\cite{HuangWang} are useful to ascertain that the loss due to decentralization cannot be worse than a certain amount. 
At the same time, since the bounds do not represent the exact gap, they do not capture the precise dependence on the system parameters and they leave open the possibility that the actual gap is asymptotically negligible.  

In comparison, in this work we derive an \emph{exact analytical formula} for this asymptotic ratio, in the case of a binary shift-in-mean detection problem involving Gaussian distributions. 
Remarkably, the analytical formula exhibits the following factorization property: the loss with respect to the optimal strategy is composed of the product of two terms, namely, $i)$ a gap due to the network structure; and $ii)$  a gap arising from the choice of the initial beliefs and/or the uneven prior assignment on the hypotheses. 
The formula reveals that, in the considered Gaussian case, the gap between the decentralized and the centralized error probabilities is \emph{irreducible}. 

Finally, we compare the most popular SL strategies in terms of error probability, revealing new results for what concerns the relative ordering of the different strategies.

\section{Background and Notation}

Social learning accounts for decentralized decision-making strategies where $K$ agents cooperate by exchanging information to reveal the true state of nature, which is one from among $H$ possible hypotheses contained in the set $\Theta = \{\theta_1, \ldots, \theta_{H}\}$. 
The true hypothesis is denoted by $\theta^{\bullet}$. 
Each agent $k$ has access to a local and private stream of data samples $\bm{x}_{k,t} \in \mathcal{X}_k$, where the subscript $t$ denotes the time instants $1,2,\ldots,$ and the bold notation denotes random quantities, which are independent over time but possibly dependent over space, i.e., across the agents. 

We start with the \emph{centralized} Bayesian system, which assumes no constraints arising from the distributed setting.

\textbf{\textit{Centralized Bayesian Strategy}.}
The Bayesian posterior based on the entire set of data available to all agents up to time $t$ amounts to computing the \emph{true} posterior probabilities 
\begin{equation}
    \bm{\mu}_t^{\star}(\theta)
    \propto \pi(\theta)\prod_{\tau=1}^t \ell_{\text{tot}}(\bm{x}_{1,\tau},\ldots,\bm{x}_{K,\tau}|\theta),
    \label{eq:truepost}
\end{equation}
where the symbol $\propto$ hides the proportionality constant needed to make $\bm{\mu}^{\star}_{t}= [\bm{\mu}^{\star}_{t}(\theta_1),\ldots,\bm{\mu}^{\star}_{t}(\theta_H)]$ a probability vector, $\pi(\theta)$ is the prior probability of $\theta$, and $\ell_{\text{tot}}(\cdot|\theta)$ is the \emph{joint} (across the agents) likelihood model. 
According to the maximum-a-posteriori probability (MAP) rule, the centralized system decides in favor of the hypothesis that maximizes $\bm{\mu}^{\star}_t(\theta)$ over $\theta\in\Theta$ and this choice minimizes the probability of error. $\hfill\square$

Unfortunately, when we move to the decentralized setting, the joint likelihood model is usually not available. 
This explains why in social learning each agent $k$ is allowed to employ only a \emph{marginal} likelihood model $\ell_k(\cdot|\theta)$. These marginal models are used to run a decentralized online algorithm of the following form: The output of agent $k$ at time $t$ is a belief vector $\bm{\mu}_{k,t} = [\bm{\mu}_{k,t}(\theta_1),\ldots,\bm{\mu}_{k,t}(\theta_H)]$, constructed based on iterative consultation steps among neighbors.
The iterative process consists of two steps: $i)$ a \emph{self-learning} step, where each agent $k$ updates individually its previous belief vector $\bm{\mu}_{k,t-1}$ by incorporating the marginal likelihood of the fresh data sample $\bm{x}_{k,t}$ to form an intermediate belief vector $\bm{\psi}_{k,t}$; $ii)$ a \emph{cooperation step}, where the intermediate beliefs are shared over the network and each agent $k$ pools the $\bm{\psi}_{j,t}$ received from its neighbors $j$ to form the output belief vector $\bm{\mu}_{k,t}$. 
Finally, agent $k$ makes its decision by picking the hypothesis to which it assigns the highest credibility, i.e., the hypothesis that maximizes $\bm{\mu}_{k,t}(\theta)$ over $\theta\in\Theta$.

Before proceeding, we introduce a basic assumption that rules out the possibility that some hypothesis is discarded \emph{ab initio} from the decision process.
\begin{assumption}[\textbf{True prior and initial beliefs}]
\label{ass:nonzerobel}
Regarding the true prior, we assume that $\pi(\theta)>0$ for all $\theta\in\Theta$. Likewise, we assume that the initial belief vectors $\mu_{k,0}$ of any SL strategy considered in the following have all nonzero entries.$\hfill\square$
\end{assumption}

\subsection{Network Graphs}
The agents communicate by relying on the \emph{network topology} that connects them, which is described by a weighted directed graph, where $i)$ the vertices are the agents; $ii)$ the edges represent the communication links; and $iii)$ the nonnegative weight $a_{jk}$ attached to the directed link from $j$ to $k$ quantifies the degree of importance that agent $k$ assigns to the information received from agent $j$. 
If the entry $a_{jk}$ is zero, then there exists no link from agent $j$ to agent $k$.
Accordingly, a (directed) path from node $j$ to node $k$ is a sequence of
links, where the first link in the sequence starts at $j$ and the last link
ends at $k$.
The neighborhood of agent $k$ is formally defined as
\begin{equation}
\mathcal{N}_k \triangleq \left\{ j: a_{jk} > 0\right\}.
\end{equation}
The communication structure can be conveniently encoded into a nonnegative $K \times K$ \emph{combination matrix} $A = [a_{jk}]$. 
In order to ensure convergence of the learning algorithm, the combination matrix $A$ must fulfill some conditions~\cite{mattaSocialLearningOpinion2024}. First of all, the sum over each column must add up to $1$.
\begin{definition}
\label{ass:leftStochasticComb}
The matrix $A$ is left-stochastic when
\begin{equation} 
\sum_{j=1}^K a_{j k} = 1,\qquad
a_{j k} \geq 0 \quad\forall j,k.
\end{equation}\hfill$\square$
\end{definition}
Note that in the special case where the rows of $A$ also add up to $1$, the matrix is said to be \emph{doubly stochastic}.

Another important attribute to achieve consistent learning is the network connectivity, which translates into the requirement of irreducibility for the combination matrix.  
\begin{definition}
The matrix $A$ is irreducible when for any agents $j$ and $k$ there exist paths that connect them in both directions.$\hfill\square$ 
\end{definition}

When $A$ is irreducible, the Perron-Frobenius theorem ensures that the matrix has a single eigenvalue equal to $1$ (albeit other eigenvalues can still have \emph{magnitude} equal to $1$). With the eigenvalue at $1$, we associate an eigenvector $v=[v_k]$, called the Perron vector, which satisfies~\cite{hornJohnson}
\begin{equation}
\sum_{j=1}^K v_j = 1, \qquad v_j > 0 \;\;\forall j.
\label{eq:perronDef}
\end{equation}
When the combination matrix is doubly stochastic, the Perron vector entries are all equal to $1/K$~\cite{mattaSocialLearningOpinion2024}. 

Finally, we introduce the concept of matrix primitivity (with reference to left-stochastic matrices). 

\begin{definition}
A left-stochastic matrix $A$ is primitive when it is irreducible and has only one eigenvalue having magnitude equal to $1$.
A sufficient condition for $A$ to be primitive is to be irreducible with at least one positive diagonal entry.$\hfill\square$
\end{definition}

It is useful to recall the following known result from matrix theory~\cite{hornJohnson}, which will be exploited in our derivations.
\begin{lemma}[\textbf{Convergence of primitive-matrix powers}]
\label{lem:classicmatlem}
Let $A$ be a left-stochastic primitive matrix with Perron vector $v$. Then,
the following bound holds:
\begin{equation}
\left|[A^t]_{jk}-v_j\right|\leq c_\lambda \cdot \lambda^t,
\label{eq:matbound}
\end{equation}
for any $0<\lambda<1$ greater than the magnitude of the second (according to a decreasing-magnitude ordering) eigenvalue of $A$. For a fixed $A$, $c_\lambda$ is a constant depending only on $\lambda$. Note that \eqref{eq:matbound} implies 
$\lim_{t\rightarrow\infty} [A^t]_{jk}=v_j$.$\hfill\square$
\end{lemma}

In summary, we will use the following assumption for the combination matrix $A$:
\begin{assumption}[\textbf{Combination matrix}]
\label{ass:primitive_combination}
The combination matrix $A$ is left-stochastic and primitive.$\hfill\square$
\end{assumption}

\subsection{Learning Algorithms}

\textbf{\textit{Traditional Social Learning}.}
In traditional social learning, the self-learning and cooperation steps are, respectively,
\begin{subequations}
\begin{align}
\bm{\psi}^{\mathrm{SL}}_{k,t}(\theta) & \propto \bm{\mu}^{\mathrm{SL}}_{k,t-1}(\theta) \, \ell_{k}(\bm{x}_{k,t}|\theta),
\label{eq:intSLStep0}
\\
\bm{\mu}^{\text{SL}}_{k,t}(\theta) & \propto \prod_{j=1}^K \left[\bm{\psi}^{\text{SL}}_{j,t}(\theta)\right]^{a_{jk}}.
\label{eq:socialLearningStep0}
\end{align}
\end{subequations}
In step \eqref{eq:intSLStep0}, each agent $k$ performs a {\em local Bayesian update} by computing the likelihood model using the fresh data sample $\bm{x}_{k,t}$, yielding an intermediate belief $\bm{\psi}^{\text{SL}}_{k,t}(\theta)$. In the second step \eqref{eq:socialLearningStep0}, agent $k$ pools the beliefs of its neighbors $j\in\mathcal{N}_k$ by using a geometric average weighted by the coefficients $a_{jk}$.$\hfill\square$

\textbf{\textit{Doubly Non-Bayesian Social Learning}.}
Observe that the cooperation step implemented by traditional social learning is non-Bayesian.
The recently proposed NB$^2$ strategy~\cite{mattaSocialLearningOpinion2024,HuangWang,bordignonSocialLearningNonBayesian2023} also implements a non-Bayesian update in the self-learning step, whence the qualification \emph{doubly non-Bayesian}.
Specifically, the NB$^2$ strategy amounts to:
\begin{subequations}
\begin{align}
\bm{\psi}^{\text{NB$^2$}}_{k,t}(\theta) & \propto \bm{\mu}^{\text{NB$^2$}}_{k,t-1}(\theta)\, \ell^{\gamma_k}_{k}(\bm{x}_{k,t}|\theta),
\label{eq:NB2intSLStep0}
\\
\bm{\mu}^{\text{NB$^2$}}_{k,t}(\theta) & \propto \prod_{j=1}^K \left[\bm{\psi}^{\text{NB$^2$}}_{j,t}(\theta)\right]^{a_{jk}},
\label{eq:NB2socialLearningStep0}
\end{align}
\end{subequations}
where $\gamma_k > 0$.
Comparing \eqref{eq:NB2intSLStep0} against \eqref{eq:intSLStep0}, we see that the likelihood is now raised to some value $\gamma_k$, which, as we will see later, has an impact on the decision performance.$\hfill\square$

\subsection{Consistent Learning}
One fundamental goal in social learning is to guarantee \emph{consistency}, that is, the capacity of \emph{all} agents to learn the truth as $t\rightarrow\infty$.

Let us denote by $\bm{\mu}_{k,t}$ a belief vector arising from some social learning strategy, not necessarily from \eqref{eq:intSLStep0}-\eqref{eq:socialLearningStep0} or \eqref{eq:NB2intSLStep0}-\eqref{eq:NB2socialLearningStep0}. 
Let
\begin{equation}
\frac 1 t \log\frac{1}{\bm{\mu}_{k,t}(\theta)}
\xrightarrow[t \rightarrow \infty]{\textnormal{a.s.}} \rho(\theta,\theta^{\bullet})>0\quad\forall \theta\neq\thetatrue,
\label{eq:RRdef}
\end{equation}
where $\xrightarrow[t \rightarrow \infty]{\textnormal{a.s.}}$ denotes almost-sure convergence as $t\rightarrow\infty$.
Condition \eqref{eq:RRdef} means that the belief about the erroneous hypothesis $\theta$ converges almost surely to $0$ at the exponential rate $\rho(\theta,\theta^{\bullet})$, which is accordingly called the \emph{rejection rate} (of $\theta$ when $\thetatrue$ is true). Equation \eqref{eq:RRdef} implies in particular that
\begin{equation}
\bm{\mu}_{k,t}(\theta^{\bullet}) \xrightarrow[t \rightarrow \infty]{\textnormal{a.s.}} 1\qquad k=1,2,\ldots,K.
\end{equation} 
Now, returning to algorithms \eqref{eq:intSLStep0}-\eqref{eq:socialLearningStep0} and \eqref{eq:NB2intSLStep0}-\eqref{eq:NB2socialLearningStep0}, when the initial beliefs at all agents and on all hypotheses are nonzero, and the combination matrix $A$ is left-stochastic and irreducible, 
it holds that~\cite[Thm. 5.2]{mattaSocialLearningOpinion2024}
\begin{equation}
\frac 1 t \log\frac{1}{\bm{\mu}^{\text{SL}}_{k,t}(\theta)}
\xrightarrow[t \rightarrow \infty]{\textnormal{a.s.}}
\sum_{j=1}^K v_j  D\left( \ell_{j,\theta^{\bullet}} || \ell_{j,\theta}\right) \triangleq 
\rho^{\text{SL}}(\theta, \theta^{\bullet}),
\label{eq:tradRejRate}
\end{equation}
where $D(\ell_{j,\theta^{\bullet}} || \ell_{j,\theta})$ denotes the Kullback-Leibler (KL) divergence between $\ell_j(\cdot|\theta^{\bullet})$ and $\ell_j(\cdot|\theta)$.
Following similar derivations to~\cite[Thm. 5.2]{mattaSocialLearningOpinion2024}, it is immediate to show that 
\begin{equation}
\frac 1 t \log\frac{1}{\bm{\mu}^{\text{NB$^2$}}_{k,t}(\theta)}
\xrightarrow[t \rightarrow \infty]{\textnormal{a.s.}}
\sum_{j=1}^K \gamma_j v_j D\left( \ell_{j,\theta^{\bullet}} || \ell_{j,\theta}\right)
\triangleq
\rho^{\text{NB$^2$}}(\theta, \theta^{\bullet}) . 
\label{eq:NBRejRate}
\end{equation}
For traditional SL and  NB$^2$, it is seen that the limiting values on the RHS of \eqref{eq:tradRejRate} and \eqref{eq:NBRejRate} stay positive (which guarantees consistency) under the assumption that, for any wrong hypothesis $\theta\neq\thetatrue$, there exists at least one agent whose marginal likelihood allows to distinguish $\theta$ from $\thetatrue$. This condition is called \emph{global identifiability}.
\begin{assumption}[\textbf{Global identifiability}]
\label{ass:globalIdent}
For each $\theta\neq\theta^{\bullet}$, at least one agent $k$ has $D(\ell_{k,\theta^{\bullet}} || \ell_{k,\theta}) > 0$. Moreover, the KL divergence between any pair of hypotheses is finite.$\hfill\square$ 
\end{assumption}

For the centralized Bayesian strategy, the condition of nonzero initial beliefs is replaced by the condition of a nonzero prior ($\pi(\theta)>0$), while the assumptions relative to the combination matrix 
are immaterial since no network is involved. 
It is known that, under global identifiability\footnote{
From the chain rule for the KL divergence~\cite{coverthomas}, global identifiability from Assumption~\ref{ass:globalIdent} implies that $D\left( \ell_{\text{tot},\theta^{\bullet}} || \ell_{\text{tot},\theta}\right)>0$. 
}~\cite{mattaSocialLearningOpinion2024},
\begin{equation}
\frac 1 t \log\frac{1}{\bm{\mu}^{\star}_{t}(\theta)}
\xrightarrow[t \rightarrow \infty]{\textnormal{a.s.}}
D\left( \ell_{\text{tot},\theta^{\bullet}} || \ell_{\text{tot},\theta}\right)\triangleq
\rho^{\star}(\theta, \theta^{\bullet})>0.
\label{eq:centRejRate}
\end{equation}

\section{Analysis of the Rejection Rate}
\label{sec:RR}

Some previous works compared the performance of different SL strategies by comparing the rejection rates of these strategies.
In the following analysis we demonstrate that the rejection rate cannot be used as a performance measure and, hence, it cannot be used for comparing algorithms.
To this end, we start by establishing the following result.

\begin{lemma}[\textbf{Any RR is achievable}]
\label{lem:anyrate}
Consider some arbitrary algorithm that generates belief vectors $\bm{\mu}_{k,t}$ and for which
\eqref{eq:RRdef} holds. 
Accordingly, for each pair of hypotheses $(\theta,\thetatrue)$, the rejection rate for this algorithm is $\rho(\theta,\thetatrue)$. Now,
for any $b>1$, there exists another belief vector $\tilde{\bm{\mu}}_{k,t}$ that achieves the larger (and, hence, superior) rejection rate $b \cdot\rho(\theta,\theta^{\bullet})$.
\end{lemma}

\begin{IEEEproof}
Consider an arbitrary real number $b>1$ and define, for all $\theta\in\Theta$, the beliefs
\begin{equation}
\tilde{\bm{\mu}}_{k,t}(\theta)\propto\bm{\mu}^{b}_{k,t}(\theta).
\label{eq:revealingtransformation}
\end{equation}
For instance, we can modify the output $\bm{\mu}_{k,t}$ of the given algorithm to obtain $\tilde{\bm{\mu}}_{k,t}$.
Using \eqref{eq:RRdef} and \eqref{eq:revealingtransformation}, we have
\begin{equation}
\frac 1 t \log\frac{\tilde{\bm{\mu}}_{k,t}(\theta^{\bullet})}{\tilde{\bm{\mu}}_{k,t}(\theta)}
=
\frac b t \log\frac{\bm{\mu}_{k,t}(\theta^{\bullet})}{\bm{\mu}_{k,t}(\theta)}
\xrightarrow[t \rightarrow \infty]{\textnormal{a.s.}} b \cdot\rho(\theta,\theta^{\bullet}).
\label{eq:intermediatelemmaproof}
\end{equation}
Since \eqref{eq:RRdef} holds by assumption, we know that $\bm{\mu}_{k,t}(\theta^{\bullet})$ converges almost surely to $1$, which in view of \eqref{eq:revealingtransformation} implies that $\tilde{\bm{\mu}}_{k,t}(\theta^{\bullet})
\xrightarrow[t \rightarrow \infty]{\textnormal{a.s.}} 1$ as well. Using this result in \eqref{eq:intermediatelemmaproof}, we conclude that
\begin{equation}
\frac 1 t \log\frac{1}{\tilde{\bm{\mu}}_{k,t}(\theta)}
\xrightarrow[t \rightarrow \infty]{\textnormal{a.s.}} b \cdot\rho(\theta,\theta^{\bullet}),
\end{equation}
i.e., $\tilde{\bm{\mu}}_{k,t}$ attains the claimed rejection rate $b\cdot \rho(\theta,\thetatrue)$.
\end{IEEEproof}

Lemma~\ref{lem:anyrate} shows that, given a decision strategy that achieves a certain rejection rate, it is always possible to modify it to achieve an arbitrarily larger rate. In other words, for a given decision-making problem, any RR is achievable. 
This suggests that the use of the RR as a performance index is generally problematic. 
In fact, we now show that several paradoxes arise when the RR is used to evaluate the decision-making performance. We start by a simple yet informative example.

\begin{example}[\textbf{Quantized data}]
\label{ex:onebit}
Let $\bm{x}_{k,t}$ be Gaussian random variables with unit variance and means $m_{\theta_1}$ and $m_{\theta_2}$ under hypotheses $\theta_1$ and $\theta_2$, respectively. These variables are assumed independent and identically distributed (iid) across space and time, i.e., with respect to indices $k$ and $t$.  
Consider a quantization threshold $\eta\in\mathbb{R}$ and define the quantized variable $\bm{q}_{k,t}=\mathbb{I}[\bm{x}_{k,t}>\eta]$,
where $\mathbb{I}$ denotes the indicator function, which is equal to $0$ if the condition defined by its argument is false, and is equal to $1$ otherwise. 
In other words, $\bm{q}_{k,t}$ is a $1$-bit version of $\bm{x}_{k,t}$. The binary variable $\bm{q}_{k,t}$ is characterized by the distribution
\begin{equation}
\mathbb{P}[\bm{q}_{k,t}=1|\theta]=Q\left(\eta-m_\theta\right),\qquad \theta\in\{\theta_1,\theta_2\},
\label{eq:quantlikelifirstexpression}
\end{equation}
where $Q$ denotes the complementary cumulative distribution function of the standard normal.
Now, let $\bm{\mu}^{(q)}_{k,t}(\theta)$ be the belief of agent $k$ at time $t$ computed with the traditional SL scheme {\eqref{eq:intSLStep0}-\eqref{eq:socialLearningStep0}} by using the likelihood of the quantized variables, denoted by $\ell^{(\text{q})}(y|\theta)=\mathbb{P}[\bm{q}_{k,t}=y|\theta]$, for $y\in\{0,1\}$ and $\theta\in\{\theta_1,\theta_2\}$. 
From \eqref{eq:tradRejRate} we have, for all $\theta\neq\theta^{\bullet}$, 
\begin{equation}
\frac 1 t \log\frac{1}{\bm{\mu}^{(q)}_{k,t}(\theta)} \xrightarrow[t \rightarrow \infty]{\textnormal{a.s.}} D\Big(\ell^{(\text{q})}_{\theta^{\bullet}}||\ell^{(\text{q})}_\theta\Big),
\label{eq:quantexRR}
\end{equation}
where the Perron vector entries $v_j$ disappear since the KL divergences are identical across the agents due to the identical distribution. 
In view of \eqref{eq:RRdef}, we conclude from  \eqref{eq:quantexRR} that the KL divergence on the RHS represents the rejection rate of the SL scheme based on the quantized variables.

Likewise, for the centralized Bayesian strategy applied to the original, unquantized data $\bm{x}_{k,t}$, from \eqref{eq:centRejRate} we get a rejection rate equal to $K \!\cdot\! D(\ell_{\theta^{\bullet}}||\ell_{\theta})$, due to the additivity of the KL divergence for independent observations. Note that the likelihoods for the unquantized variables have been denoted by $\ell_\theta$ and $\ell_{\thetatrue}$, in place of $\ell_{k,\theta}$ and $\ell_{k,\thetatrue}$, in view of the identical distribution across the agents..

Now, the data processing inequality for the KL divergence implies~\cite{KLDPI}
\begin{equation}
D(\ell_{\theta^{\bullet}}||\ell_{\theta})\geq D(\ell^{(\text{q})}_{\theta^{\bullet}}||\ell^{(\text{q})}_\theta).
\end{equation} 
On the other hand, from Lemma \ref{lem:anyrate} we know that we can modify the scheme based on the quantized variables by constructing a belief
\begin{equation}
\tilde{\bm{\mu}}^{(q)}_{k,t}(\theta)\propto \left[\bm{\mu}^{(q)}_{k,t}(\theta)\right]^b,
\label{eq:ex1modq}
\end{equation}
with $b$ chosen so as to achieve a rejection rate 
\begin{equation}
b \cdot  D(\ell^{(\text{q})}_{\theta^{\bullet}}||\ell^{(\text{q})}_\theta) \geq K \cdot  D(\ell_{\theta^{\bullet}}||\ell_{\theta}).    
\end{equation}
As a result, if the rejection rate is taken as a performance proxy, then one can construct a decentralized scheme based on quantized variables that outperforms the centralized Bayesian scheme based on the original unquantized variables. This conclusion is rather suspicious.$\hfill\square$
\end{example}

\subsection{Paradoxes Arising from the RR}

{\bf \emph{Paradox 1: Information reduction does not affect RR}.} The analysis presented in Example~\ref{ex:onebit} implies that, given an arbitrary decision strategy $\mathbb{A}$ based on some input data $\bm{x}_{k,t}$, one can construct a decision strategy $\mathbb{B}$, based on one-bit quantized data, which outperforms $\mathbb{A}$. This goes against one fundamental principle of statistical inference and information theory, namely, the data processing inequality. Applying that principle, data quantization entails (but for special/trivial cases) a reduction of information, and an inferential strategy based on the reduced data cannot outperform the best strategy based on unquantized data.$\hfill\square$

{\bf \emph{Paradox 2: Discarding information improves RR}.} 
Consider a decision-making problem where all agents are able to distinguish $\thetatrue$ from $\theta\neq \thetatrue$, with one particular agent $k^\star$ having the highest discriminative capacity, namely, an agent such that
\begin{equation}
D(\ell_{k^\star,\thetatrue}||\ell_{k^\star,\theta})>
D(\ell_{k,\thetatrue}||\ell_{k,\theta})\qquad \textnormal{for any $k\neq k^\star$}.
\end{equation}
Then, the rejection rate in \eqref{eq:tradRejRate} is maximized with the choice $v_{k^\star}=1$, which basically means that the agent $k^{\star}$ can work in isolation, with the data of all other agents being discarded. 
This strategy goes against another fundamental principle of statistical inference and information theory, according to which an optimal decision strategy never discards informative data. 
For example, consider a binary detection problem with the observations being Gaussian with different means under the different hypotheses. Assume also that the observations of different agents have different variances, with the most informative agent featuring the smallest variance. Then, the optimal decision strategy would scale the observations with weights inversely proportional to the variances, but would never discard any observation.$\hfill\square$

{\bf \emph{Paradox 3: The exact posterior is not the best belief}.} 
The exact posterior probability $\bm{\mu^\star}_{t}(\theta)$ represents the most faithful and complete statistical description of the hypotheses given the data. 
However, Lemma \ref{lem:anyrate} guarantees that we can always find a belief that attains a rejection rate larger than the centralized rejection rate $\rho^{\star}(\theta,\theta^{\bullet})$ from \eqref{eq:centRejRate}. 
This implies that, in terms of rejection rate, the exact posterior is outperformed by (infinitely many) strategies.$\hfill\square$ 

In summary, the analysis in this section demonstrates that the rejection rate is not a reliable index for the decision-making performance. 
As a matter of fact, in decision theory the RR is not considered as a performance measure.

\section{Error Probability vs. Rejection Rate}

Recalling that the decision is made by taking the hypothesis that maximizes the belief, the error probability of agent $k$ \emph{given that the true hypothesis is $\theta^{\bullet}$} is given by
\begin{equation}
\label{eq:conderr}
    p_{k,t}(\theta^{\bullet})=\mathbb{P}_{\theta^{\bullet}}[\bm{\mu}_{k,t}(\theta^{\bullet})\leq \bm{\mu}_{k,t}(\theta) \textnormal{ for at least one $\theta\neq \theta^{\bullet}$}],
\end{equation}
where the notation $\mathbb{P}_{\theta^{\bullet}}$ signifies that the true hypothesis is $\theta^{\bullet}$.
Then, the total error probability of agent $k$ is obtained by averaging over the prior distribution $\pi$, yielding
\begin{equation}
    p_{k,t}=\sum_{\theta^{\bullet}\in\Theta}\pi(\theta^{\bullet}) p_{k,t}(\theta^{\bullet}).
    \label{eq:toterrprob}
\end{equation} 
By inspecting the definition of the rejection rate in \eqref{eq:RRdef}, it is easy to fall into the misconception that an increased rejection rate corresponds to a reduced error probability. However, this conclusion is false, as we now show.

\begin{figure*}[ht]
    \centering        \includegraphics[width=0.325\linewidth]{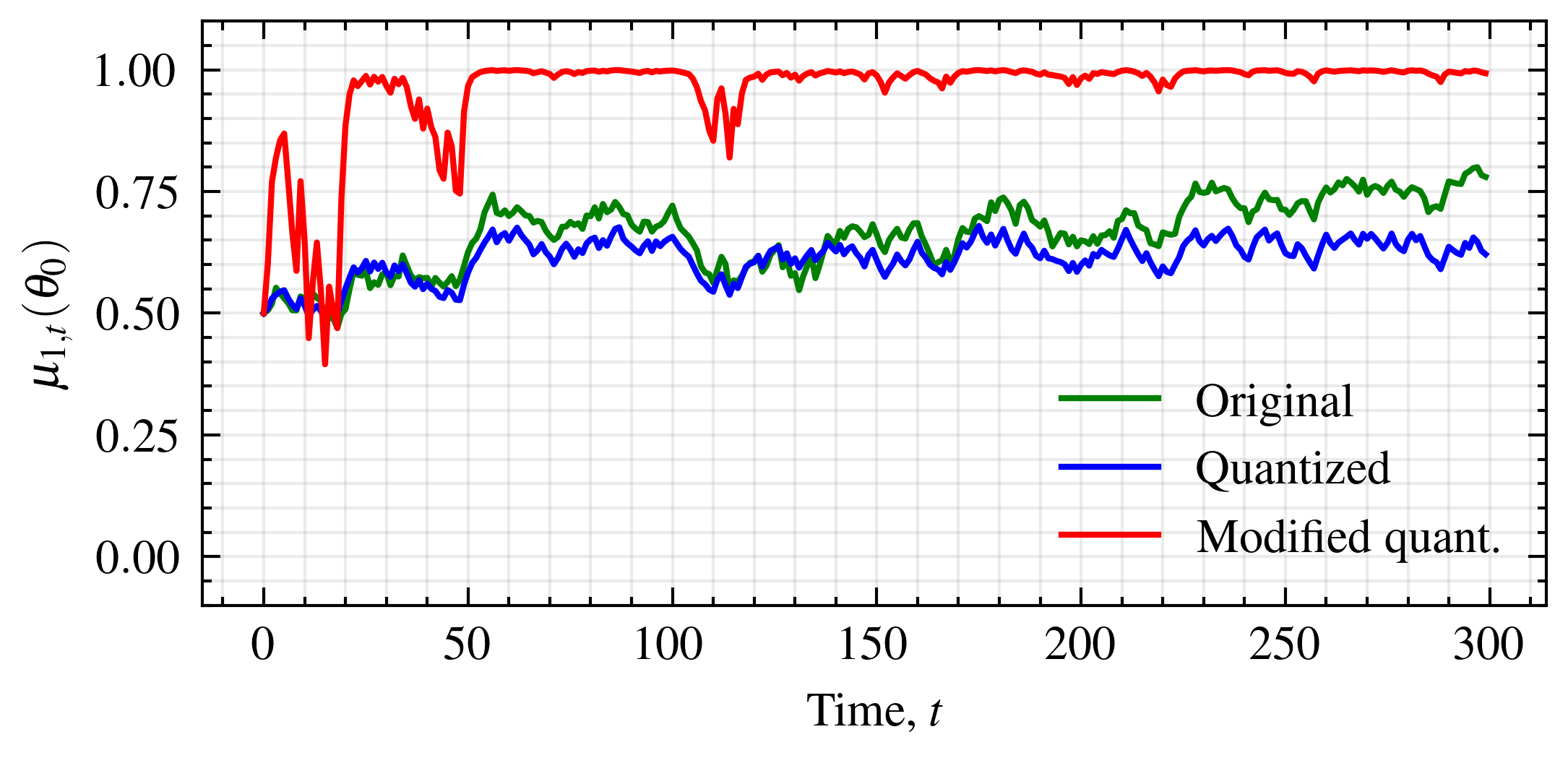}
    \hfill
    \includegraphics[width=0.325\linewidth]{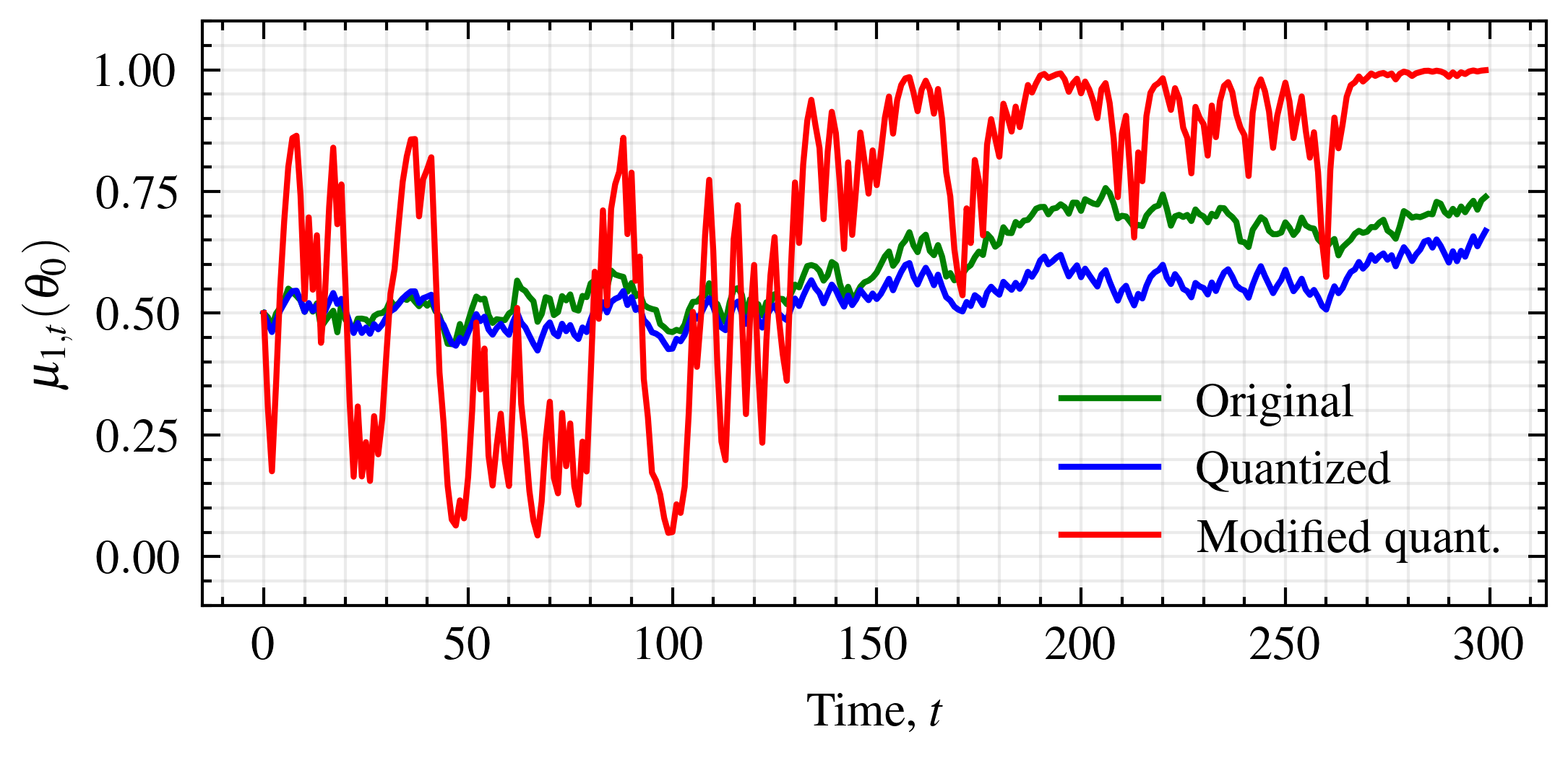}
    \hfill
    \includegraphics[width=0.325\linewidth]{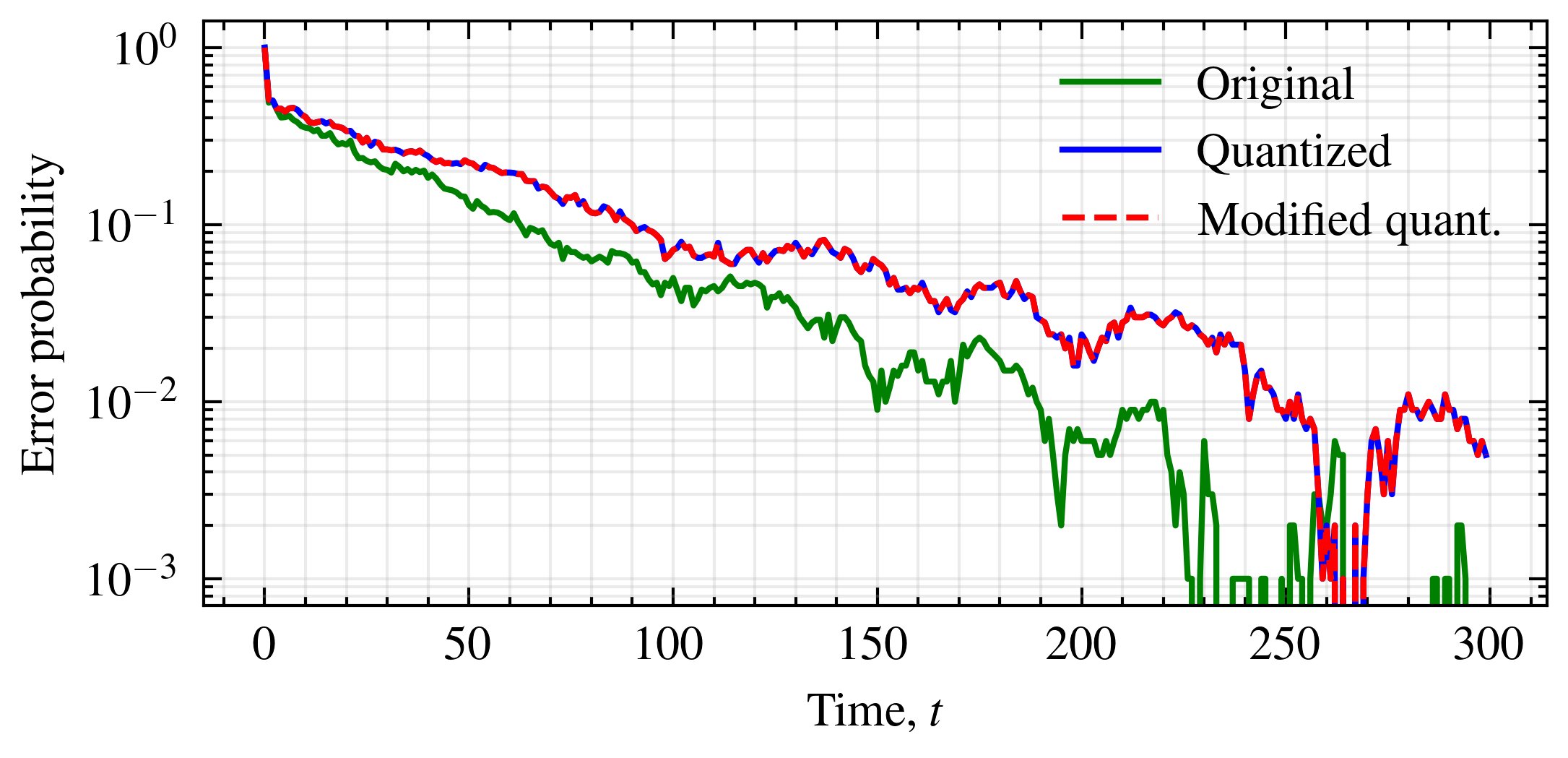}
    \caption{
    \emph{Rejection rate vs. error probability.}
    Traditional SL behavior in the setting of Example~\ref{ex:onebit}, for the network topology shown in the inset panel of Fig.~\ref{fig:th1unifbel}. (\emph{Left}). One realization of the belief $\bm{\mu}_{1,t}(\theta^{\bullet})$ where, after a random time instant $\bm{t}_0 \approx 25$, the beliefs of the compared strategies follow the ordering dictated by the rejection rates. (\emph{Middle}). A different realization of the belief $\bm{\mu}_{1,t}(\theta^{\bullet})$, where the strategies follow the ordering of the rejection rates after a longer time, i.e., $\bm{t}_0 \approx 130$. Note that the modified quantized implementation is characterized by a more significant variability of the beliefs with respect to traditional SL. (\emph{Right}). Error probability curves obtained by averaging over multiple realizations. The quantized implementations share the same error probability curves in view of \eqref{eq:equivEvents}, but they feature a higher error probability compared to traditional SL. }
    \label{fig:RRvserrprob}    
\end{figure*}

Consider two social learning strategies $1$ and $2$, yielding two belief vectors $\bm{\mu}'_{k,t}$ and $\bm{\mu}''_{k,t}$, whose respective rejection rates fulfill the ordering relation $\rho''(\theta,\thetatrue)>\rho'(\theta,\thetatrue)$ for all $\theta\neq\thetatrue$.
If the beliefs $\bm{\mu}'_{k,t}$ and $\bm{\mu}''_{k,t}$ were deterministic, then from \eqref{eq:RRdef} we would conclude that, from a certain \emph{deterministic} time instant $t_0$ onward, $\bm{\mu}''_{k,t}(\theta)<\bm{\mu}'_{k,t}(\theta)$ for all $\theta\neq\thetatrue$. Under this condition, whenever strategy 1 selects correctly $\thetatrue$, so does strategy 2. We would then conclude that the error probability associated with strategy 2 cannot be higher than the error probability associated with strategy 1. 
However, we now show that this conclusion is invalidated by the fact that the beliefs are in fact \emph{random functions}, with the convergence in \eqref{eq:RRdef} being a \emph{stochasic convergence}, specifically, an almost-sure convergence. 

The fact that we deal with almost-sure convergence implies that, for (almost) any realization, there exists a \emph{random} time instant $\bm{t}_0$ such that, for $t\geq \bm{t}_0$, $\bm{\mu}''_{k,t}(\theta)<\bm{\mu}'_{k,t}(\theta)$ for all $\theta\neq\thetatrue$. In other words, the time epoch $\bm{t}_0$ depends on the realization, which changes completely the perspective in terms of error probability.
To illustrate this point better, it is convenient to refer back to Example \ref{ex:onebit}. 
Consider first the MAP rule based on \eqref{eq:truepost} (i.e., on the original unquantized data $\bm{x}_{k,t}$). 
We know that this rule minimizes the error probability. 
Consider then the two beliefs introduced in Example \ref{ex:onebit}: the belief $\bm{\mu}^{(\text{q})}_{k,t}$ corresponding to traditional SL applied to \emph{quantized} data and its modified version $\tilde{\bm{\mu}}^{(\text{q})}_{k,t}$ designed to outperform the MAP rule in terms of rejection rate. 

In the first two panels of Fig.~\ref{fig:RRvserrprob}, we display two realizations of the beliefs of the aforementioned three strategies, evaluated at $\thetatrue$. In both realizations, after sufficient time we see that the ordering of the beliefs follows the ordering dictated by the rejection rates, hence, the strategy with the highest belief is the modified quantized strategy. However, in the realization of the middle panel, the time for reaching this condition is significantly larger than the time observed in the leftmost panel. Moreover, in the middle panel we see that, before convergence, the belief of the modified quantized strategy exhibits a much greater variability. 
In particular, there are many time instants where the modified quantized strategy makes mistakes (the belief about the true hypothesis is even close to zero), whereas the MAP rule does not. 

Averaging over multiple realizations, we obtain the error probability curves displayed in the rightmost panel, where we observe the following behavior. 
First, despite the modified quantized rule exhibiting a rejection rate larger than the quantized strategy, both strategies have the same error probability. 
This is a consequence of the following equivalence between events: 
\begin{equation}
\left\{\tilde{\bm{\mu}}^{(\text{q})}_{k,t}(\theta^{\bullet})\leq\tilde{\bm{\mu}}^{(\text{q})}_{k,t}(\theta)\right\}
\Leftrightarrow
\left\{\bm{\mu}^{(\text{q})}_{k,t}(\theta^{\bullet})\leq\bm{\mu}^{(\text{q})}_{k,t}(\theta)\right\},
\label{eq:equivEvents}
\end{equation}
which follows from \eqref{eq:ex1modq}.
Second, the MAP rule has the smallest error probability. 
This is because the quantized strategies are based on transformed data (with loss of information) and the MAP rule is optimal.

In summary, it is wrong to assume that increasing the RR reduces the error probability. 
The RR represents indeed the \emph{limiting average value} of a decision statistic (the log belief ratio). Since it is a limiting average, it cannot account for the random variability of the beliefs and, consequently, for the number of mistaken decisions made over time. 
This aspect is readily captured by referring to a classic example from decision-making theory, namely, a binary detection problem. Consider two decision statistics used to solve the detection problem. The goal is to establish which decision statistic performs better. 
Comparing the \emph{averages} of the decision statistics to compare the two strategies would be wrong, as it is well-known that the expected values do not carry useful information about the decision performance, which requires to consider at least the variance, or more complete descriptors such as the error probability.

\section{Sensitivity to the Network Connectivity and the Belief Initialization}
\label{sec:netconnandprior}

In decentralized learning, it is often the case that, in the long term, the performance is equalized across the agents and the initial system state does not play a role.
This is true, for example, for decentralized estimation/regression~\cite{sayednewbooks}. 
This is also true in decentralized decision-making, with reference to 
the so-called \emph{error exponent}, which represents the exponential rate at which the error probability converges to $0$ as $t \rightarrow \infty$ (\emph{not} the convergence rate of the beliefs quantified by the rejection rate). Formally,
\begin{equation}
p_{k,t}= \exp\big(-E \, t + o(t) \big),
\label{eq:LD}
\end{equation}
where $o(t)$ is a function such that $\lim_{t\rightarrow\infty} o(t)/t=0$. For the strategies examined in this work, it has been shown that the error exponent $E$ is equalized across the agents and unaffected by the prior distribution and the initial beliefs~\cite{mattaSocialLearningOpinion2024}.
The exact expression for the exponent $E$ can be computed by resorting to the \emph{theory of large deviations} \cite{demboLargeDeviationsTechniques2009, denhollanderLargeDeviations2000}, which
has been applied before to distributed detection systems, e.g., in problems involving sensor networks \cite{SungTongPoor,Kar2018}, in decentralized binary detection problems \cite{MouraLDnonGauss, MouraDIperformance}, and more specifically in social learning \cite{mattaSocialLearningOpinion2024,lalithaSocialLearningDistributed2018}. 

Note that the term $o(t)$ in \eqref{eq:LD} can be a constant or can diverge to $\infty$ as $t\rightarrow\infty$.
As a result, two strategies that share the same error exponent can still have different error probabilities because of differences in these higher-order terms. 
This implies that, even if the exponents are equalized across the agents and are not influenced by the prior and the initialization, this need not be the case for the actual error probability. 

Deriving a general analytical expression to highlight this result is a formidable task. The analysis is viable for the Gaussian case, examined in the next theorem.
The goal of the theorem is to compare the error probability $p_{k,t}$, relative to agent $k$ in a social learning strategy, against the optimal error, namely, the probability $p^\star_t$ attained by the centralized Bayesian system. 
Specifically, we will evaluate the limit, as $t\rightarrow\infty$, of the ratio $p_{k,t}/p^\star_t$.
To make the comparison meaningful, we need to consider social learning implementations that match the error exponent of the optimal centralized strategy, otherwise the limit would trivially be equal to $+\infty$. 
It was shown in~\cite{mattaSocialLearningOpinion2024, bordignonSocialLearningNonBayesian2023,HuangWang} that, under Assumptions~\ref{ass:nonzerobel}, \ref{ass:primitive_combination}, and~\ref{ass:globalIdent}, and under statistical independence over time and across the agents, the same error exponent as the optimal centralized strategy is attained: $i)$ by traditional SL when the combination matrix is doubly stochastic; and $ii)$ for left-stochastic matrices, by the NB$^2$ strategy with weights $\gamma_k=1/v_k$.

\begin{theorem}[\textbf{Error probabilities for binary Gaussian problems}]
\label{th:HyperCos}
Let Assumptions~\ref{ass:nonzerobel},~\ref{ass:primitive_combination}, and~\ref{ass:globalIdent} hold.
Let $\bm{x}_{k,t}$ be Gaussian variables, independent across space and time (i.e., with respect to indices $k$ and $t$) with means $m_k(\theta_1)$ and $m_k(\theta_2)$ under hypotheses $\theta_1$ and $\theta_2$, respectively, and with variance $s^2_k$ under both hypotheses. 
Let
\begin{equation}
\Delta_k =\frac{\big(m_k(\theta_2)-m_k(\theta_1)\big)^2}{2s^2_k}
\label{eq:KLdeltadef}
\end{equation} 
be the KL divergence between the two hypotheses \cite{coverthomas} and define the error term
\begin{equation}
\varepsilon_k\triangleq \frac 1 4 \sum_{\tau=1}^\infty\sum_{j=1}^K \left(
\frac{[A^{\tau}]_{jk}}{v_j} - 1
\right)^2\,\Delta_j.
\label{eq:overallneterr}
\end{equation}
Let $p^{\star}_t$ denote the error probability relative to the optimal centralized Bayesian system employing the MAP rule.
Likewise, let $p_{k,t}$ be the error probability of agent $k$ at time $t$, either relative to the traditional SL strategy with doubly stochastic $A$ or to the NB$^2$ strategy with $\gamma_k=1/v_k$ (and left-stochastic $A$). 
Furthermore, let 
\begin{equation}
\xi\triangleq \log\frac{\pi(\theta_1)}{\pi(\theta_2)},\qquad \xi_{\mathrm{net}}\triangleq 
\sum_{j=1}^K v_j
\log\frac{\mu_{j,0}(\theta_1)}{\mu_{j,0}(\theta_2)},
\label{eq:xixinet}
\end{equation}
where $\pi(\theta)$ is the true prior probability of observing $\theta$, while $\mu_{j,0}(\theta)$ is the initial belief about $\theta$ employed by agent $j$.
Then, we have the following result:
\begin{equation} 
\lim_{t\rightarrow\infty}\frac{p_{k,t}}{p^\star_t}
=
\exp\left(
\varepsilon_k
\right)\,
\cosh
\left(
\frac
{\xi-\xi_{\mathrm{net}}/\alpha}
{2}
\right),
\label{eq:th1main}
\end{equation}
where $\alpha=1/K$ for the traditional SL strategy with doubly stochastic $A$ and $\alpha=1$ for the NB$^2$ strategy.
\end{theorem}
\begin{IEEEproof}
See Appendix~\ref{app:1}.
\end{IEEEproof}

The structure of \eqref{eq:th1main} suggests a sharp interpretation of the discrepancy between the decentralized and the centralized error probabilities.
In fact, formula \eqref{eq:th1main} reveals that the ratio $p_{k,t}/p^\star_t$, for large $t$, is given by the product of two terms of exponential type. 
Preliminarily, we observe that both terms are $\geq 1$, as we expect since $p^\star_t$ is the best, i.e., the minimum, error probability.  
Let us examine the two terms more closely.

\subsection{Network Error}
The first term on the RHS of \eqref{eq:th1main}, $\exp(\varepsilon_k)$, is determined by the network error $\varepsilon_k$ in \eqref{eq:overallneterr}. 
Examining \eqref{eq:overallneterr}, we conclude that this error is invariably present in any social learning strategy. This is because, excluding trivial cases like a fully connected topology, the convergence of the matrix power entries $[A^\tau]_{jk}$ to the Perron vector entries $v_j$ is only reached asymptotically as $\tau\rightarrow\infty$. 
As a result, we have $\exp(\varepsilon_k)=1$ for the fully connected case, but in general we will have $\exp(\varepsilon_k)>1$. 
This means that, contrary to a widespread belief, \emph{the decentralized decision-making strategy exhibits a gap, due to imperfect network connectivity, with respect to the optimal centralized strategy. This gap does not disappear as $t\rightarrow\infty$}. 
Actually, in the literature of decentralized learning it is often implied that the optimal centralized performance is attained asymptotically. 
This turns out to be true for decentralized estimation problems, where the natural performance measure is the mean-square-error, and the impact of the network error term actually disappears. However, we have now established that this is not true for decision-making problems. 

Moreover, we see that the error $\varepsilon_k$ depends on $k$, which means that different agents may have different errors. Examining \eqref{eq:overallneterr}, this is due to the fact that the terms $[A^t]_{jk}$ depend on $k$, revealing another distinctive behavior. 
While for other decentralized learning problems the asymptotic performance is equalized across the agents, this is no longer the case for social learning, where \emph{different agents exhibit different behavior even for large $t$}.\footnote{Discrepancies among the agents' performance were observed in the context of \emph{adaptive} social learning~\cite{MattaTSIPNexactasy,mattaSocialLearningOpinion2024}. However, adaptive social learning is a different problem, because the error probability does not vanish with time, and the asymptotic analysis is performed in terms of an adaptation parameter that converges to zero.}  

The following corollary provides a uniform (i.e., independent of $k$) upper bound on $\varepsilon_k$ that summarizes an average dependence on the overall network structure.
\begin{corollary}[\textbf{Bound on the network error}]
\label{cor:1}
The error $\varepsilon_k$ defined by \eqref{eq:overallneterr} obeys the following bound:
\begin{equation}
|\varepsilon_k|\leq \frac{c_\lambda^2}{4} \frac{\lambda^2}{1 - \lambda^2} \sum_{j=1}^K \frac{\Delta_j}{v_j^2},
\end{equation}
where the quantities $\lambda$ and $c_\lambda$ are introduced in Lemma~\ref{lem:classicmatlem}.
\end{corollary}
\begin{IEEEproof}
The proof is obtained by applying \eqref{eq:matbound} to \eqref{eq:overallneterr}.
\end{IEEEproof}
We see from the definition of $\lambda$ in Lemma~\ref{lem:classicmatlem} that the bound is essentially determined by the second largest-magnitude eigenvalue of the combination matrix $A$. The effect of this eigenvalue on the convergence of matrices is classically encountered in the mixing-time of Markov chains and also in consensus or diffusion problems~\cite{BoydFastMixing, Consensus, sayednewbooks, mattaSocialLearningOpinion2024}. It is known that smaller second eigenvalue magnitudes favor the diffusion of information.   

It is also useful to examine the structure of $\varepsilon_k$ for the two considered decentralized strategies, namely, traditional SL and NB$^2$. 
Recall that for traditional SL we require the combination matrix to be doubly stochastic. 
Accordingly, we rewrite \eqref{eq:overallneterr} for the case $v_j=1/K$, obtaining
\begin{equation}
\varepsilon_k\triangleq \frac 1 4 \sum_{\tau=1}^\infty\sum_{j=1}^K \left(
K[A^{\tau}]_{jk} - 1
\right)^2\,\Delta_j.
\end{equation} 
Now, if NB$^2$ is run with a doubly stochastic matrix, we will get the same expression, that is, the network error term has the same impact for traditional SL and NB$^2$. 
Instead, when the combination matrix is not doubly stochastic, the comparison is meaningless because we know that traditional SL would not achieve the optimal error exponent.

\subsection{Agents' Prior Knowledge and Initialization}
\label{sec:priorinitnewsec}
Let us move on to examine the second term appearing in \eqref{eq:th1main}, namely, the hyperbolic cosine term. This term is $\geq 1$ and is $1$ only when the argument is $0$. 
We need to consider the interplay between the agents' prior knowledge and the initialization of their beliefs. Now, to choose their initial beliefs $\mu_{k,0}(\theta)$, the agents rely on their \emph{individual} prior knowledge. 
When they start with different prior convictions or, more generally, when their knowledge is insufficient to determine the true prior $\pi(\theta)$, the quantities $\xi$ and $\xi_{\mathrm{net}}/\alpha$ from \eqref{eq:xixinet} will be different. As a result, the hyperbolic cosine in \eqref{eq:th1main} will be strictly greater than $1$, which  corresponds to a loss in terms of decentralized error probability.
Thus, our analysis leads to the following conclusion: \emph{Contrary to a widespread conviction, the agents' prior knowledge and the choice of the initial beliefs have a critical effect on the error probability performance, which does not disappear even as $t\rightarrow\infty$.} 
This marks a fundamental difference with respect to other decentralized learning problems (such as regression or estimation), where the role of the initial state is asymptotically immaterial.

Consider next the situation where all agents are able to retrieve the \emph{true prior $\pi(\theta)$} and accordingly choose their initial beliefs as
\begin{equation}
\mu_{k,0}(\theta)=\pi(\theta)\;\;\;\text{$\forall k$ and $\theta$.}
\label{eq:initbelequaltotrueprior}
\end{equation}
In view of \eqref{eq:xixinet}, this choice implies that $\xi_{\mathrm{net}}=\xi$.
Examining the hyperbolic cosine term in \eqref{eq:th1main}, and recalling that $\alpha=1$ for the NB$^2$ strategy, we reach the following conclusion: \emph{when the agents' prior knowledge is exact and they set their initial beliefs equal to the true prior, the error probability for the NB$^2$ strategy is only affected by the network gap $\exp(\varepsilon_k)$.} 

Let us switch to the analysis of traditional SL (with doubly stochastic matrices), for which we have $\alpha=1/K$. Under the assignment \eqref{eq:initbelequaltotrueprior}, if the true prior is uniform, we have $\xi=\xi_{\mathrm{net}}=0$, and the hyperbolic cosine term disappears.
If the prior is not uniform, the hyperbolic cosine term becomes
\begin{equation}
\cosh\left(\frac{(K-1)\,\xi}{2}\right)
\label{eq:specialHyperCosTrad}
\end{equation}
and is strictly greater than $1$. 
We conclude that, \emph{
when the agents' prior knowledge is exact and they set their initial beliefs equal to the true prior, the error probability for traditional SL is affected by both the network gap $\exp(\varepsilon_k)$ and the initialization gap \eqref{eq:specialHyperCosTrad} (unless the true prior is uniform).
}

We now provide an interpretation for this behavior. First of all, observe from \eqref{eq:th1main} that the roles of the network connectivity and the initial beliefs are made independent by the product form, namely, the network effect is embodied into the term $\exp(\varepsilon_k)$, while the initialization effect appear in the hyperbolic cosine. 
Thus, to provide a qualitative interpretation for the role of the initialization alone, it is instructive to consider a fully connected network, where the network error can be easily made zero by the assignment $a_{jk}=1/K$ for all $j$ and $k$. 
Note that this uniform assignment also implies that $[A^t]_{jk}=1/K$ for all $t$. By developing the recursion in \eqref{eq:intSLStep0}-\eqref{eq:socialLearningStep0}, we obtain
\begin{equation}
\log\frac
{\bm{\mu}^{\mathrm{SL}}_{k,t}(\theta_1)}
{\bm{\mu}^{\mathrm{SL}}_{k,t}(\theta_2)}=
\frac 1 K \sum_{\tau=1}^t\sum_{j=1}^K 
\log\frac{\ell_j(\bm{x}_{j,t-\tau+1}|\theta_1)}{\ell_j(\bm{x}_{j,t-\tau+1}|\theta_2)}
+
\log\frac{\pi(\theta_1)}{\pi(\theta_2)}.
\label{eq:fullconntrad}
\end{equation}
Likewise, unfolding the recursion in \eqref{eq:NB2intSLStep0}-\eqref{eq:NB2socialLearningStep0} with the choice $\gamma_j=K$, we obtain
\begin{equation}
\log\frac
{\bm{\mu}^{\text{NB$^2$}}_{k,t}(\theta_1)}
{\bm{\mu}^{\text{NB$^2$}}_{k,t}(\theta_2)}=
\sum_{\tau=1}^t\sum_{j=1}^K 
\log\frac{\ell_j(\bm{x}_{j,t-\tau+1}|\theta_1)}{\ell_j(\bm{x}_{j,t-\tau+1}|\theta_2)}
+
\log\frac{\pi(\theta_1)}{\pi(\theta_2)}.
\label{eq:fullconnB2}
\end{equation}
Finally, for the optimal centralized system in \eqref{eq:truepost} we can write
\begin{equation}
\log\frac{\bm{\mu}^\star_{t}(\theta_1)}{\bm{\mu}^\star_{t}(\theta_2)}=
\sum_{\tau=1}^t\sum_{j=1}^K 
\log\frac{\ell_j(\bm{x}_{j,t-\tau+1}|\theta_1)}{\ell_j(\bm{x}_{j,t-\tau+1}|\theta_2)}
+
\log\frac{\pi(\theta_1)}{\pi(\theta_2)}.
\label{eq:optcen}
\end{equation}
We see that, while \eqref{eq:fullconnB2} and \eqref{eq:optcen} coincide, the traditional SL recursion from \eqref{eq:fullconntrad} is different, because the term depending on the log likelihoods is divided by $K$. This enhances the role of the prior information and results in the performance worsening discussed so far. 

At the same time, the comparison between \eqref{eq:fullconntrad} and \eqref{eq:optcen} suggests that the error probabilities for traditional SL and the optimal Bayesian strategy would be equivalent if the term $\log(\pi(\theta_1)/\pi(\theta_2))$ in \eqref{eq:fullconntrad} was also divided by $K$. This scaling can be achieved by the following initial belief assignment:
\begin{equation}
\mu^{\mathrm{SL}}_{k,0}(\theta)\propto\big(\pi(\theta)\big)^{1/K}\;\;\;\text{$\forall k$ and $\theta$},
\label{eq:initbelequaltoscaledtrueprior}
\end{equation}
which, by developing the recursion in \eqref{eq:intSLStep0}-\eqref{eq:socialLearningStep0}, yields
\begin{equation}
K\log
\frac
{\bm{\mu}^{\mathrm{SL}}_{k,t}(\theta_1)}
{\bm{\mu}^{\mathrm{SL}}_{k,t}(\theta_2)}
=
\sum_{\tau=1}^t\sum_{j=1}^K 
\log\frac{\ell_j(\bm{x}_{j,t-\tau+1}|\theta_1)}{\ell_j(\bm{x}_{j,t-\tau+1}|\theta_2)}
+
\log\frac{\pi(\theta_1)}{\pi(\theta_2)}.
\label{eq:fullconntradnewinitnewnew}
\end{equation}
Since the log belief ratios in \eqref{eq:optcen} and \eqref{eq:fullconntradnewinitnewnew} differ by a constant, they lead to the same error probabilities. 
In other words, the effect of the initialization is canceled by using a \emph{wrong prior!} This result, which might look suspicious, can be explained as follows. It is readily verified that the beliefs corresponding to \eqref{eq:fullconntradnewinitnewnew} and \eqref{eq:optcen} are related by the following identity:
\begin{equation}
\bm{\mu}^{\mathrm{SL}}_{k,t}(\theta)\propto\big(
\bm{\mu}^{\star}_t(\theta)\big)^{1/K}
\;\;\;\text{$\forall k$ and $\theta$}.
\label{eq:intermsofposter}
\end{equation}
While this identity implies the equivalence of the error probabilities obtained with the belief vectors $\bm{\mu}^{\mathrm{SL}}_{k,t}$ and $\bm{\mu}^{\star}_t$, it also shows that traditional SL (even over a fully connected network) does not track faithfully the optimal Bayesian posterior.

Before concluding this section, we observe from \eqref{eq:xixinet} that the assignment in \eqref{eq:initbelequaltoscaledtrueprior} for traditional SL yields $\xi=\xi_{\mathrm{net}}\,K=\xi_{\mathrm{net}}/\alpha$. Substituting this equality into \eqref{eq:th1main}, we see that the hyperbolic cosine term is annihilated. We remark that this property does not require in any manner the condition of full connectivity that we have used before to provide a qualitative argument. Therefore, we have reached another interesting conclusion: \emph{When the agents' prior knowledge is exact, the error probability for traditional SL is only affected by the network gap $\exp(\varepsilon_k)$ provided that the agents initialize their beliefs with the modified version of the true prior distribution described by \eqref{eq:initbelequaltoscaledtrueprior}.} 

In summary, our analysis revealed that the initial belief assignment affects the agents' error probabilities. Therefore, to maximize their performance, the agents should possess a reliable prior knowledge and set carefully their initial beliefs based on the available prior distribution.

\section{Numerical Examples}
\label{sec:illex}

In the present section, we will examine the performance of the considered social learning strategies with different network topologies, combination policies, prior distributions and initial belief assignments. Regarding the topologies, there is usually no particular reason for the agents to ignore their own data. This means that one assumes that each agent has a self-loop. In the following, we make this assumption even if it is not required by our theorems.
Moreover, when the network graph reflects the links over a communication network, it is often the case that it is \emph{undirected}. This is because full-duplex communication links are usually available.  For this reason, in most applications the graph is undirected and  the combination matrix $A$ is doubly stochastic, since there exist many policies (e.g., the Metropolis and the Laplacian rules) that allow to obtain a doubly stochastic matrix over an undirected graph~\cite{mattaSocialLearningOpinion2024,sayednewbooks}. 
However, there may be applications where the graph is directed. Note that the assumptions in our theorems do not require undirected graphs. Accordingly, in the following we will consider both undirected and directed graphs.

\subsection{Theorem~\ref{th:HyperCos} at Work}
\label{sec:th1Exp}

We consider a Gaussian shift-in-mean decision problem, where the (Gaussian) observations $\bm{x}_{k,t}$ are iid with unit variance and with means equal to $m_k(\theta_1)=0$ and $m_k(\theta_2)=0.1$. 
The network is composed of $10$ agents linked according to the undirected topology displayed in the inset panel of Fig.~\ref{fig:th1unifbel}. On top of this topology we apply the Metropolis policy, which is defined, for $j\neq k$, by the assignment~\cite{mattaSocialLearningOpinion2024,sayednewbooks}
\begin{equation}
a_{jk}=
\begin{cases}
\dfrac{1}{\max\{|\mathcal{N}_j|,|\mathcal{N}_k|\}}\;&\textnormal{if $j\in\mathcal{N}_k$},\\
0\;&\textnormal{otherwise},
\end{cases}
\end{equation}
with the normalization condition $a_{kk}=1- \sum_{j\neq k}a_{jk}$. Since the resulting matrix $A$ is symmetric and left-stochastic, it is also doubly stochastic.
In this example, all agents are initialized with uniform beliefs and the true prior $\pi(\theta)$ is uniform as well.

\begin{figure}[t]
    \centering
        \includegraphics[width=0.5\linewidth]{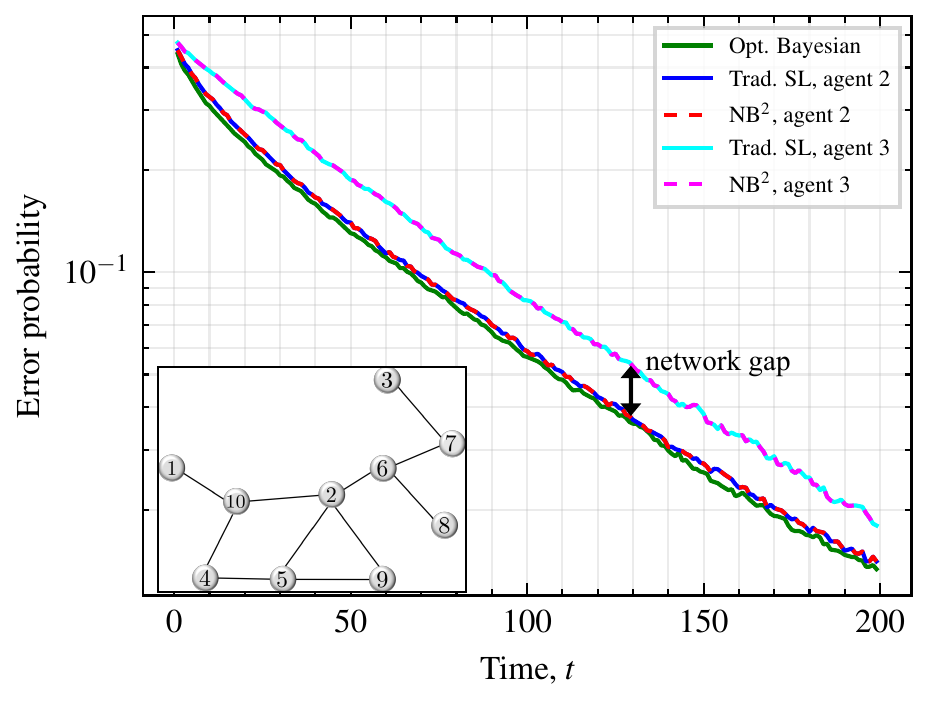}
    \caption{ 
    \emph{Network error.} Error probabilities of the considered SL strategies for the decision-making problem in Sec.~\ref{sec:th1Exp}. In this experiment, the prior and the initial beliefs are uniform. For the decentralized implementations, the agents are connected according to the network topology shown in the inset panel (here the edges are depicted with lines and no arrows, since they are \emph{undirected}). The error probabilities of agents $2$ and $3$ are shown, since these two agents represent a \emph{central} and a \emph{peripheral} agent. This choice highlights the effect of the network topology on the performance. 
    The arrows emphasize the network error $\exp(\varepsilon_k)$ in \eqref{eq:th1main}. 
    In agreement with Theorem~\ref{th:HyperCos}, traditional SL and NB$^2$ feature the same curves because of the uniform initialization. 
    All the curves are obtained averaging over $10^4$ Montecarlo runs.}
    \label{fig:th1unifbel}
\end{figure}

\begin{table*}[t]
\centering
\resizebox{\linewidth}{!}{
\begin{tabular}{|c|c|c|c|c|c|c|c|c|c|c|}
    \hline
     & Agent $1$ & \cellcolor{gray!15}Agent $2$ & \cellcolor{gray!15}Agent $3$ & Agent $4$ & Agent $5$ & Agent $6$ & Agent $7$ & Agent $8$ & Agent $9$ & Agent $10$\\
    \hline
    Undirected (Metropolis) & $0.300$ & \cellcolor{gray!15}$0.038$ & \cellcolor{gray!15}$0.382$ & $0.135$ & $0.097$ & $0.096$ & $0.241$ & $0.280$ & $0.130$ & $0.116$ \\
    \hline
    Directed (Metropolis) & $0.199$ & \cellcolor{gray!15}$0.117$ & \cellcolor{gray!15}$0.309$ & $0.156$ & $0.231$ & $0.072$ & $0.187$ & $0.225$ & $0.337$ & $0.094$ \\
    \hline
    Directed (Laplacian) & $0.671$ & \cellcolor{gray!15}$0.254$ & \cellcolor{gray!15}$0.934$ & $0.497$ & $0.626$ & $0.254$ & $0.538$ & $0.642$ & $0.726$ & $0.298$ \\
    \hline
\end{tabular}
}
\vspace{1ex}
\caption{Network errors \eqref{eq:overallneterr} for: $i)$ the undirected graph in Fig.~\ref{fig:th1unifbel}, equipped with the Metropolis rule; $ii)$ the directed graph in Fig.~\ref{fig:directTopoMetropolis}, equipped with the Metropolis rule; and $iii)$ the directed graph in Fig.~\ref{fig:directTopoLaplacian}, equipped with the Laplacian rule. The shaded columns highlight the agents on which we focus in Figs.~\ref{fig:th1unifbel},~\ref{fig:directTopoMetropolis}, and~\ref{fig:directTopoLaplacian}.}
\label{tab:neterrs}
\end{table*}

Figure~\ref{fig:th1unifbel} shows the error probabilities of agents $2$ and $3$ for the three strategies considered in Theorem~\ref{th:HyperCos}, i.e.,  traditional SL \eqref{eq:intSLStep0}-\eqref{eq:socialLearningStep0}, NB$^2$ \eqref{eq:NB2intSLStep0}-\eqref{eq:NB2socialLearningStep0}, and the centralized MAP rule based on the true posterior \eqref{eq:truepost}.

We observe the following trends. 
First, in the logarithmic scale (on the vertical axis) the error probability curves are nearly parallel. This matches perfectly the theoretical results on the error exponent presented at the beginning of Sec.~\ref{sec:netconnandprior}, where we explained that, with doubly stochastic matrices and observations independent over time and across the agents, both traditional SL and NB$^2$ (with weights $\gamma_k=1/v_k$) attain the optimal error exponent achieved by the MAP rule. 
We note in passing that, with uniform initial beliefs, the curves of traditional SL and NB$^2$ are identical. This is obvious because the log belief ratios of the two strategies differ by a constant term $K$, which does not affect the error probability. 

The second main observation is that, albeit parallel, the curves pertaining to the decentralized strategy do not coincide with the curve for the optimal system. This is predicted by Theorem~\ref{th:HyperCos}. In particular, the theorem states that, for the case of uniform beliefs, the asymptotic gap between the error probability $p_{k,t}$ of agent $k$ and the optimal error probability $p^\star_t$ is given by the factor $\exp(\varepsilon_k)$, which depends on the specific connectivity of agent $k$. 
Using formula \eqref{eq:overallneterr}, we have computed the network errors $\varepsilon_k$ for each agent $k$, as reported in Table~\ref{tab:neterrs}, first row. We see that the network error for agent $2$ is very small, which explains why its probability curve approaches more closely the optimal curve. 
The gap due to the network error corresponding to agent $3$ is instead non-negligible, and is emphasized by the arrows in the figure.

\begin{figure}[t]
    \centering
        \includegraphics[width=0.5\columnwidth]{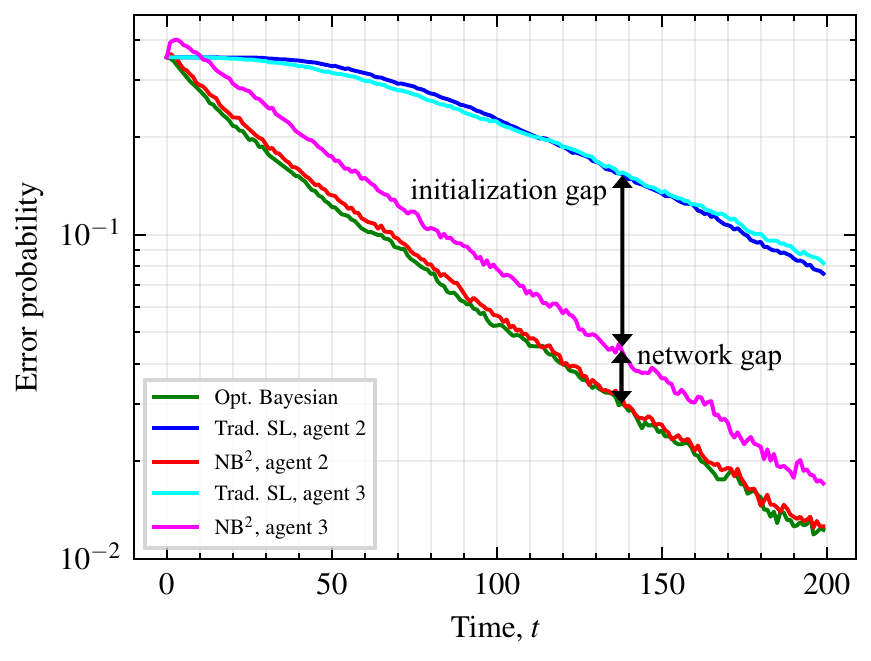}
    \caption{\emph{Belief initialization.} Same setup used in Fig.~\ref{fig:th1unifbel}, but for the prior and the initial beliefs, which are now set to $[0.65, 0.35]$.  
    Comparing against Fig.~\ref{fig:th1unifbel}, and in agreement with Theorem~\ref{th:HyperCos}, we see that now the error probability for traditional SL pays an additional error term quantified (in the logarithmic scale adopted in the figure) by $\log_{10}\cosh(\,(K-1)\,\xi/2\,)$, due to the non-uniform initialization, yielding an overall gap equal to $\varepsilon_k\,\log_{10} e+\log_{10}\cosh(\,(K-1)\,\xi/2\,)$. The arrows in the figure emphasize the two additive 
    terms of this gap.}
    \label{fig:th1nonunifbel}
\end{figure}

Let us move on to examine the additional impact of the agents' prior knowledge and initial belief assignment. 
To this end, in Fig~\ref{fig:th1nonunifbel}, we consider the same setting used in Fig.~\ref{fig:th1unifbel}, but for one variation: the true prior is no longer uniform, in particular, we have $\pi(\theta_1)=0.65$. All agents use this prior as initial belief.
We see again that the error probability curves stay nearly parallel for sufficiently large $t$, i.e., they feature the same error exponent. 
However, while the relative behavior of the NB$^2$ and the optimal strategy is the same as before, now a difference emerges for what concerns traditional SL. The error probability curves of the latter scheme are about one order of magnitude higher than the curve pertaining to NB$^2$ and the optimal Bayesian classifier. 
This is predicted by Theorem~\ref{th:HyperCos}, which reveals that, in the presence of a non-uniform prior distribution, the traditional SL strategy initialized with the correct prior pays an additional loss (with respect to the $\exp(\varepsilon_k)$ term due to imperfect connectivity). This loss is quantified by the hyperbolic cosine in \eqref{eq:specialHyperCosTrad}, yielding, in the logarithmic scale adopted in Fig.~\ref{fig:th1nonunifbel}, an overall gap equal to $\varepsilon_k\,\log_{10}e + \log_{10}\cosh(\,(K-1)\,\xi/2\,)$. 
The arrows in the figure emphasize the two components of this gap, namely, the contribution due to the network error and the contribution due to initialization.
As explained in Sec.~\ref{sec:priorinitnewsec}, to compensate for the initialization gap, it is necessary to initialize traditional SL with the wrong prior in \eqref{eq:initbelequaltoscaledtrueprior}. In this case (recall that we are using a doubly stochastic combination matrix and, hence, NB$^2$ uses constant weights $\gamma_k=K$), the error probability curves for traditional SL would be superimposed to the error probability curves for NB$^2$.

\begin{figure}[t]
    \centering
        \includegraphics[width=0.5\columnwidth]{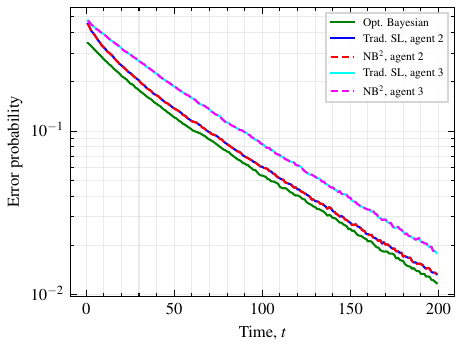}
    \caption{\emph{Lack of prior knowledge.} Same setup used in Fig.~\ref{fig:th1unifbel}, but for the true prior, which is now set to $[0.65, 0.35]$. However, the agents do not know the true prior and continue to use uniform initial beliefs as in Fig.~\ref{fig:th1unifbel}. 
    }
    \label{fig:th1nonunifbelunifbel}
\end{figure}

Finally, we want to examine the effect of the agents' prior knowledge.
To this end, in Fig~\ref{fig:th1nonunifbelunifbel}, we modify again the setting used in Fig.~\ref{fig:th1unifbel} by imposing a non-uniform prior $\pi(\theta_1)=0.65$. However, we now assume that, due to the lack of reliable prior knowledge, the agents continue to work under a uniform assignment for their initial beliefs (which, since we are considering a doubly stochastic matrix, implies that the error probabilities for traditional SL and NB$^2$ are equal). Under this scenario, we see from \eqref{eq:xixinet} that the term $\xi_{\mathrm{net}}$ (which depends on the initial beliefs) is zero, while the term $\xi$ (which depends on the true prior) is not, resulting into the gap $\cosh(\xi/2)$.
Comparing Fig.~\ref{fig:th1nonunifbelunifbel} against Fig.~\ref{fig:th1unifbel}, we see that the gap between the decentralized schemes and the optimal Bayesian strategy increases. This is particularly visible for agent $2$, whose error probability curve in Fig.~\ref{fig:th1unifbel} was in practice coincident with the optimal curve.

\subsection{Role of the Topology and the Combination Policy}
The error term \eqref{eq:overallneterr} captures the effect of the network on the error probability of each agent $k$. 
In this section we examine the behavior of this term in connection with the topology and the combination policy. 
To this aim, it is useful to consider a uniform prior distribution and uniform initial beliefs, so that the hyperbolic cosine term in \eqref{eq:th1main} is nullified and the network effect emerges more clearly.

In Fig.~\ref{fig:th1unifbel}, we considered an undirected graph and a Metropolis combination policy, yielding a doubly stochastic matrix $A$. We start by modifying the topology in the following manner: first we remove the link between agents $2$ and $5$. Then, in the loop that involves agents $2,  4, 5,  9, 10$, we turn the original undirected edges into directed  edges, in a counterclockwise direction.  The resulting graph is therefore \emph{directed} and \emph{less connected} than the one used in Fig.~\ref{fig:th1unifbel}. On top of this new graph, we continue to use the Metropolis combination policy.

In Fig.~\ref{fig:directTopoMetropolis}, we report the error probability curves pertaining to agents $2$ and $3$, namely, the same agents on which we focused in Fig.~\ref{fig:th1unifbel}. 
First of all, we see that the curves pertaining to traditional SL and NB$^2$ are no longer superimposed, which admits the following explanation. 
We know that the Metropolis rule yields by construction a left-stochastic matrix. However,  over a \emph{directed} graph there is no guarantee that the resulting combination matrix is doubly stochastic. As a matter of fact, we have verified that the matrix $A$ obtained by applying the Metropolis rule to the topology in Fig.~\ref{fig:directTopoMetropolis} is not doubly stochastic. 
We know that, under non-doubly-stochastic matrices, NB$^2$ attains the optimal error exponent, while traditional SL does not.
Therefore, we expect that traditional social learning is asymptotically worse than NB$^2$. This matches the behavior observed in Fig.~\ref{fig:directTopoMetropolis}, even if for agent $2$ the difference is minimal.

\begin{figure}[t]
    \centering
    \includegraphics[width=0.5\columnwidth]{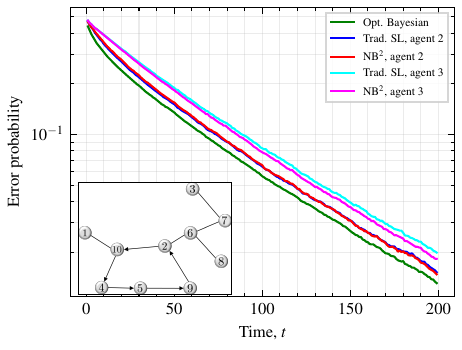}
    \caption{
    \emph{Impact of directed graphs.} Same setup used in Fig.~\ref{fig:th1unifbel}, but for the network topology, which is described by the directed graph shown in the inset panel.
    }
    \label{fig:directTopoMetropolis}
\end{figure}

We focus next on the NB$^2$ strategy, which, as just said, is optimal at the first exponential order.
Thus, we need to examine the higher-order terms affecting the error probability. Since the true prior and the initial beliefs are uniform, the higher-order terms are only determined by the network errors $\varepsilon_k$ in \eqref{eq:overallneterr}, which are reported in Table~\ref{tab:neterrs}, second row. 
In particular, comparing the errors for agent $2$ in the new topology and in the previous topology (first row in the table), we see that the network error for the new graph is increased. This is also confirmed by the observation that the error probability of agent $2$ in Fig.~\ref{fig:directTopoMetropolis} is  worse than the error probability of the optimal system, while in Fig.~\ref{fig:th1unifbel} it was nearly optimal. 
This worsening looks reasonable, in light of the fact that the new graph is sparser. However, we will now explain that the sparsity in the graph is not the only factor that influences the network error. To this end, we switch to the analysis of agent $3$. From Table~\ref{tab:neterrs}, we see that now the network error associated with the undirected graph is lower, which means that the sparser directed graph does not worsen the performance of agent $3$. This behavior calls for a closer look at
the structure of the error in \eqref{eq:overallneterr}. This error is surely influenced by the topology (since the support graph of the combination matrix $A$ is dictated by the topology), but also by: $i)$ the specific values of the combination matrix entries (determined by the combination policy); $ii)$ the Perron vector entries $v_j$; and $iii)$ the individual agents' KL divergences $\Delta_j$, which in this particular case do not matter since they are set equal for all $j$. The considered example shows that the interplay among these factors is complex, and the mere connectivity of the graph is not enough to draw conclusions as regards the performance of each individual agent. Notably, the closed-form expression for the network error \eqref{eq:overallneterr} is able to capture the overall effect of these factors on the error probabilities. In particular, our examples highlighted that the network errors are in general variable across the agents, and that a change in the topology may have contrasting effects on different agents. This variability is embodied in formula \eqref{eq:overallneterr}. We remark that a traditional approach in distributed optimization and decision-making is to encode the role of the topology in the second largest magnitude eigenvalue as we did, for example, in Corollary~\ref{cor:1}. However, this approach is insufficient to capture the differences emerging across the agents, and the more detailed representation provided by formula \eqref{eq:overallneterr} is necessary.

\begin{figure}[t]
    \centering
    \includegraphics[width=0.5\columnwidth]{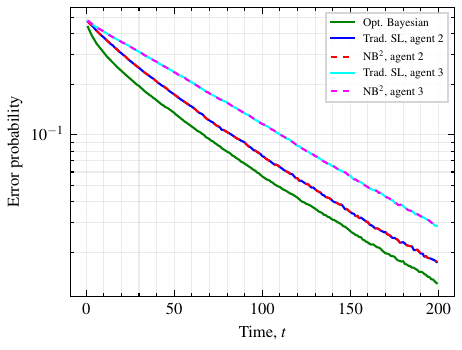}
    \caption{\emph{Role of the combination policy.} 
    Same setup used in Fig.~\ref{fig:directTopoMetropolis}, but for the combination policy, which is now the Laplacian rule in \eqref{eq:Laplace} with $\phi=0.5$.
    }
    \label{fig:directTopoLaplacian}
\end{figure}

Let us now focus on the effect of the combination policy. Over the considered directed topology, we replace the Metropolis with the Laplacian rule, which is defined, for $j\neq k$, by the assignment~\cite{mattaSocialLearningOpinion2024,sayednewbooks}
\begin{equation}
a_{jk}=
\begin{cases}
\dfrac{\phi}{\max\limits_{i=1,2,\ldots,K}|\mathcal{N}_i|}\;&\textnormal{if $j\in\mathcal{N}_k$},\\
0\;&\textnormal{otherwise},
\end{cases}
\label{eq:Laplace}
\end{equation}
with $a_{kk}=1- \sum_{j\neq k}a_{jk}$. 
In particular, in the considered experiments, we set $\phi=0.5$. We have verified that, for the topology in Fig.~\ref{fig:directTopoMetropolis}, the combination matrix obtained with the Laplacian policy is doubly stochastic. Figure~\ref{fig:directTopoLaplacian} reports the error probability curves pertaining to this new policy. Traditional SL and NB$^2$ exhibit the same performance, as it must be since $A$ is doubly stochastic. 
Comparing Fig.~\ref{fig:directTopoLaplacian} against Fig.~\ref{fig:directTopoMetropolis}, we see that the performance of both agents $2$ and $3$ is worsened. 
This is confirmed by the evaluation of the network errors, see the third row in Table~\ref{tab:neterrs}.

We remark that the conducted experiments are meant to provide hints about the impact that the particular topology or combination policy may have on the performance, and how the obtained formula is able to capture this impact. 
On the other hand, the experiments cannot provide any conclusion regarding the optimization of the combination policy or the topology to maximize the performance. For example, by optimizing the choice of the free parameter $\phi$ in the Laplacian rule, we can improve the performance obtained in the particular experiment shown in Fig.~\ref{fig:directTopoLaplacian}.

\subsection{Comparing Strategies}
\label{sec:compStratEx}

In this section we compare the three strategies considered so far (namely, traditional SL, NB$^2$, and the optimal centralized MAP) against the following two alternative strategies proposed in the literature: the arithmetic averaging strategy, which in the self-learning  step implements the classic Bayesian update \eqref{eq:intSLStep0}, while in the pooling step \eqref{eq:socialLearningStep0} replaces the weighted geometric average of the intermediate beliefs with their arithmetic average \cite{zhaoLearningSocialNetworks2012, mattaSocialLearningOpinion2024}; and the recent approach proposed in \cite{mitraNewApproachDistributed2021}, which uses a more conservative communication rule where each agent $k$ shares the minimum among the belief updated locally by itself and the beliefs of its neighbors. 
For ease of presentation, in the following we will refer to the latter strategy as the min-belief strategy.

To go beyond the binary setting considered by Theorem~\ref{th:HyperCos}, we
compare the aforementioned five strategies under a ternary classification problem where $\bm{x}_{k,t}$ is a $5$-dimensional vector whose entries are iid unit-variance Gaussian variables with means $0$, $0.05$, and $0.1$ under hypotheses $\theta_1$, $\theta_2$, and $\theta_3$. 
The observations are iid over time and across the agents. 

\begin{figure*}[t]
    \centering
        \includegraphics[width = 0.25\linewidth]{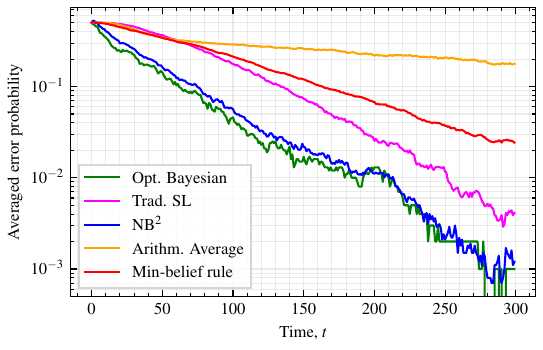}
    \hfill
        \includegraphics[width = 0.25\linewidth]{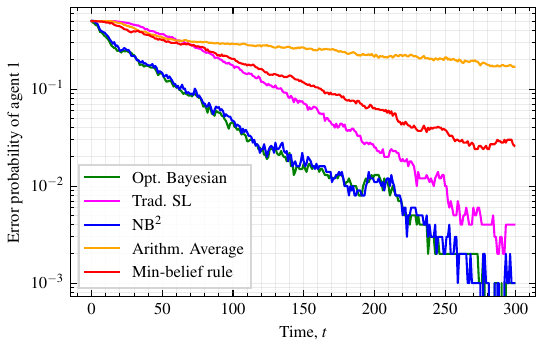}
    \hfill
        \includegraphics[width = 0.25\linewidth]{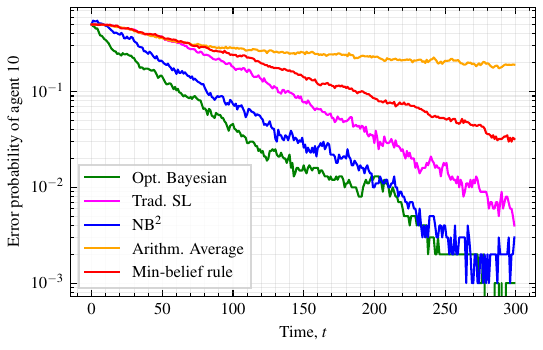}
    \hfill
    \begin{minipage}[h]{0.1\linewidth}
        \vspace{-2.75cm}
        \centering
        \includegraphics[width = \columnwidth]{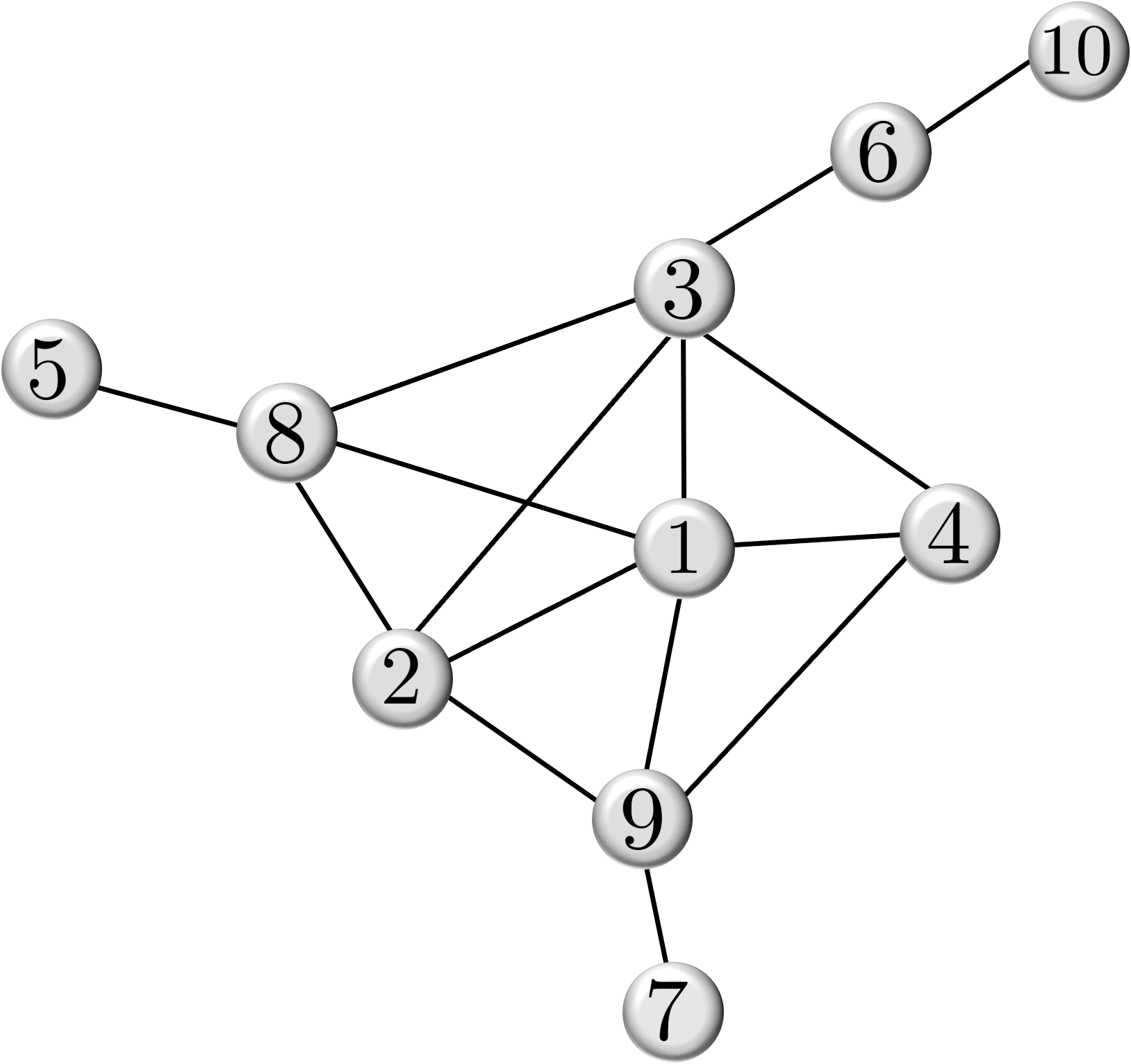}
    \end{minipage}
    \caption{\emph{Comparing strategies.} Error probabilities attained by different SL strategies, for the decision-making problem detailed in Sec.~\ref{sec:compStratEx}. For the decentralized implementations, the agents are connected according to the network topology 
    shown in the rightmost panel. The leftmost panel shows the error probability averaged across the agents, while the second and third panels show the error probabilities of agents $1$ and $10$ (one central and one peripheral agent, respectively). 
    All the curves are obtained averaging over $10^3$ Montecarlo runs.
    }
\label{fig:nonunifgentopolog}
\end{figure*}

\begin{figure*}[t]
    \centering
        \includegraphics[width = 0.25\linewidth]{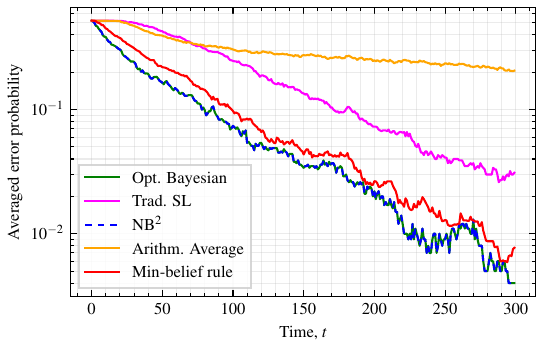}
    \hfill
        \includegraphics[width = 0.25\linewidth]{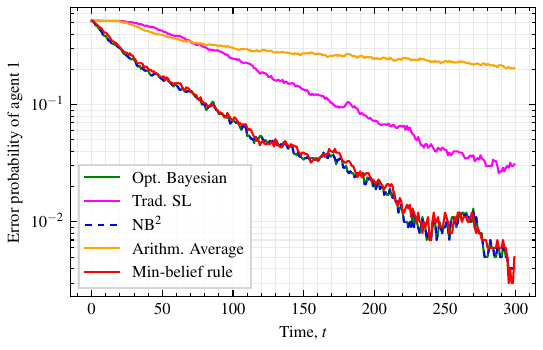}
    \hfill
        \includegraphics[width = 0.25\linewidth]{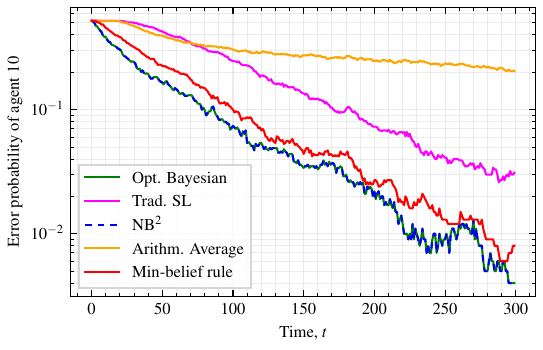}
    \hfill
    \begin{minipage}[h]{0.1\linewidth}
        \vspace{-2.75cm}
        \centering
        \includegraphics[width = \columnwidth]{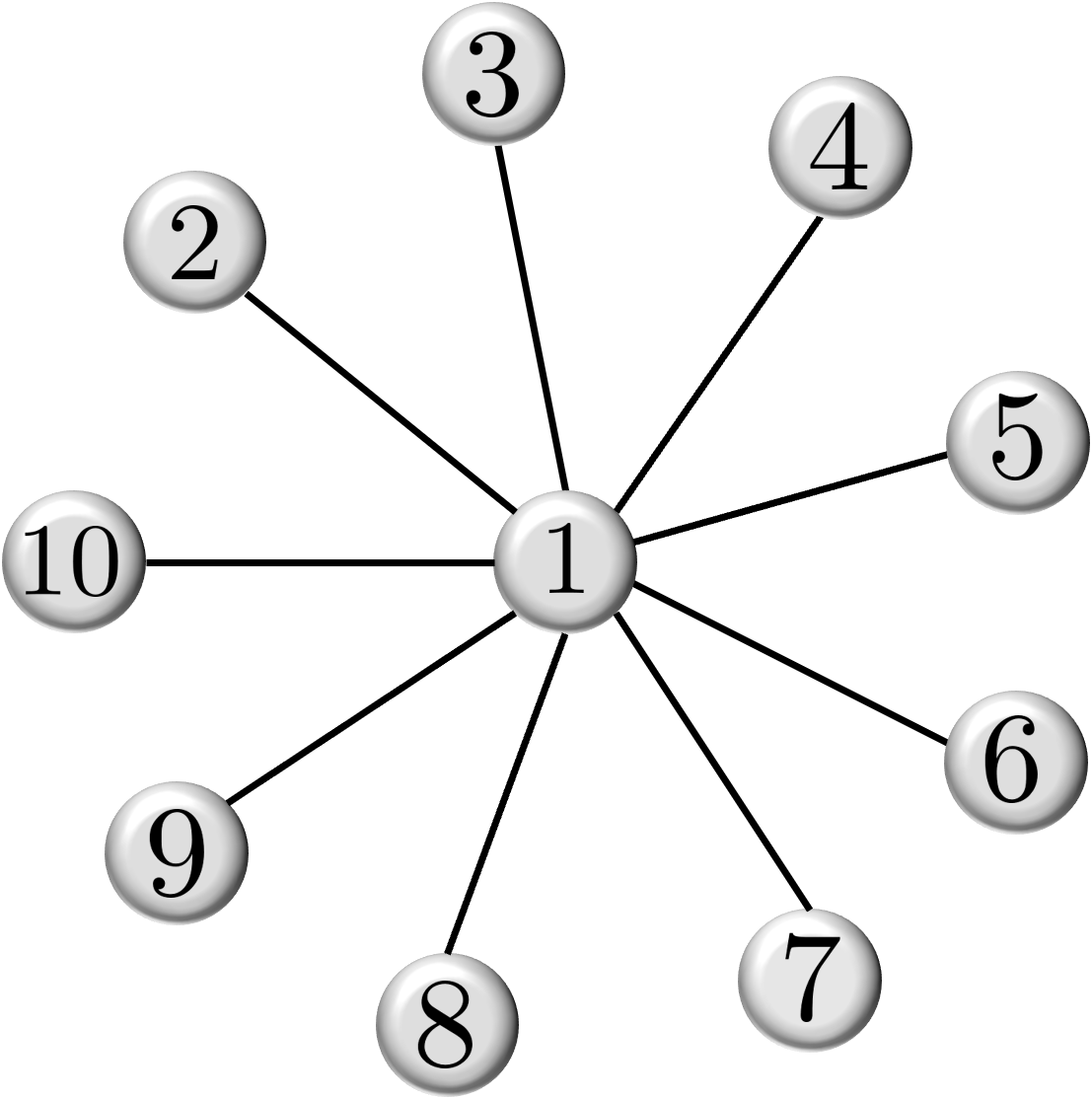}
    \end{minipage}
    \caption{
    \emph{Comparing strategies.}
    Same setup used in Fig.~\ref{fig:nonunifgentopolog}, but for the network topology, which is now the star topology shown in the rightmost panel. 
    }
\label{fig:topostar}
\end{figure*}

In Fig.~\ref{fig:nonunifgentopolog}, we show the error probability curves for the described setup. Specifically, the leftmost panel shows the error probability curves \emph{averaged across the agents}; the second and third panels show the curves pertaining to agents $1$ and $10$, respectively; the rightmost panel depicts the network topology, on top of which we designed a Metropolis matrix. The prior is given by $\pi=[0.2, 0.3, 0.5]$, and the initial beliefs are all set equal to the prior.

The behavior of traditional SL, NB$^2$, and the optimal centralized strategy are similar to the behavior observed in the previous examples. We also recall that if traditional SL was initialized with the modified prior in \eqref{eq:initbelequaltoscaledtrueprior}, its error probability would coincide with that of NB$^2$.

Let us examine how the two other strategies perform.
Arithmetic averaging performs significantly worse. In particular, its error probabilities exhibit a much slower convergence rate, that is, the error exponent is significantly higher than the common exponent characterizing traditional SL, NB$^2$, and the optimal centralized MAP. 
Also for the min-belief strategy, the error exponent is worse than the optimal exponent and, essentially for all time instants, the error probability curves are higher than the error probability curves of traditional SL, NB$^2$, and the centralized scheme.  

In order to gain further insight into the behavior of the min-belief rule, in Fig.~\ref{fig:topostar} we repeat the experiment over a star topology and an identifiability setup adapted to this topology, as considered in~\cite{mitraNewApproachDistributed2021}. 
The center of the star topology corresponds to an agent that has an identifiable decision model, whereas all the remaining agents are completely uninformative. Thus, the optimal centralized decision rule coincides with the MAP rule implemented locally by the informative agent alone. On the other hand, all the uninformative agents can easily mirror the likelihood ratios of the informative agent, since they are connected to it over the star topology. By construction, traditional SL and NB$^2$ are able to mimic this protocol. In particular, over a star topology the effect of the network error can be annihilated by using the Metropolis combination policy.\footnote{This is verified by noting that $i)$ in the network error \eqref{eq:overallneterr} only the term relative to the informative agent $1$ is nonzero; and $ii)$ for the Metropolis policy applied to the star topology, $[A^\tau]_{1k}=v_1=1/K$ for all $\tau$ and $k$, implying $\varepsilon_k=0$ for all $k$.}
Once a null network error is guaranteed, from Theorem~\ref{th:HyperCos} we conclude that the NB$^2$ decision rule becomes equal to the MAP rule and that the same equivalence would hold for traditional SL under the prior assignment in  \eqref{eq:initbelequaltoscaledtrueprior}.

For the min-belief rule, while the second panel shows that the informative agent attains the optimal performance, the third panel reveals that this is no longer the case for the uninformative agents. This means that the min-belief rule introduces an additional error also in this simple ad-hoc scenario. In contrast, in~\cite{mitraNewApproachDistributed2021} the superiority of the min-belief rule was claimed. However, this claim was based on the rejection rate, which as thoroughly discussed in Sec.~\ref{sec:RR}, is not a valid performance measure. 

In summary, the analysis conducted in this section highlights that it is fundamental to choose an appropriate metric to compare the decision-making strategies. The error probability is one such metric. 
The best decentralized strategy is NB$^2$, which allows to manage non-doubly-stochastic combination matrices;
under doubly stochastic matrices, traditional SL exhibits the same error probability as NB$^2$ under the prior assignment in  \eqref{eq:initbelequaltoscaledtrueprior}.
Other strategies, such as arithmetic averaging or min-belief, exhibit a non-negligible performance gap. 
To avoid misunderstanding, we remark that the fact that by using the rejection rate or the error probability one gets different results \emph{cannot} be justified by saying that one is using two different criteria. More simply, the conclusions obtained when using the rejection rate are not reliable since we have proven that the rejection rate cannot be considered a meaningful performance criterion.

\section{Conclusion}
We addressed the problem of characterizing the performance of decentralized decision-making, a.k.a. social learning. 
First, we showed that the rejection rate is not an appropriate performance metric.
We turned then to the classic decision performance index, the error probability. 
From the asymptotic analysis of this probability in terms of error exponent, we know that traditional SL and a recent strategy called NB$^2$ attain the optimal centralized exponent for the case of independent agents.
However, we showed that this asymptotic equivalence hides important effects due to decentralization. 
The take-away is that, differently from other decentralized learning problems (such as estimation/regression), in decision-making the individual agents' connectivity and initialization do matter and are not washed out in the long run.   
The comparison of different strategies in terms of error probability revealed interesting behavior and confirmed the predictions of our theoretical analysis, according to which NB$^2$ is the best strategy.
Several useful extensions are possible. 
For example, Theorem~\ref{th:HyperCos} revealed the presence of a network error that depends on the powers of the combination matrix entries and also, for NB$^2$, on the Perron vector. It would be interesting to examine whether the network error can be minimized by optimizing the choice of the network topology and/or the combination policy.
Another interesting topic would be the optimization of the SL strategies to account for the dependence across the agents. 
Finally, a challenging open problem is the assessment of a closed-form relation, holding beyond the Gaussian problem covered by Theorem~\ref{th:HyperCos}, for the agents' error probabilities. 
We are currently pursuing this research line by exploiting the framework of \emph{exact asymptotics}~\cite{demboLargeDeviationsTechniques2009, MattaTSIPNexactasy}.

\appendices
\section{Auxiliary Results}
\begin{lemma}[\textbf{Convergent Network Errors}]
\label{lem:neterr}
Under Assumption~\ref{ass:primitive_combination}, the following series are convergent:
\begin{align}
\varepsilon_k'&\triangleq \sum_{\tau=1}^\infty
\sum_{j=1}^K 
\left(
[A^{\tau}]_{jk} - v_j
\right)\,\gamma_j\,\Delta_j,
\label{eq:epsprime}
\\
\varepsilon_{k}''&\triangleq
4 \sum_{\tau=1}^\infty\sum_{j=1}^K \left(
[A^{\tau}]_{jk} - v_j
\right) \,v_j\,\gamma_j^2\,\Delta_j
+
2 \sum_{\tau=1}^\infty\sum_{j=1}^K \left(
[A^{\tau}]_{jk} - v_j
\right)^2\,\gamma_j^2\, \Delta_j.
\label{eq:epsprimeprime}
\end{align}
\end{lemma}
\begin{IEEEproof}
Exploiting \eqref{eq:matbound}, the triangle inequality, and the  geometric series, we can write
\begin{equation}
|\varepsilon_k'|
\leq c_\lambda  \sum_{\tau=1}^\infty \lambda^\tau 
\sum_{j=1}^K \gamma_j\,\Delta_j=
\frac{c_\lambda\,\lambda}{1-\lambda}\sum_{j=1}^K \gamma_j\,\Delta_j,
\label{eq:1bound}
\end{equation}
which guarantees the convergence of the series \eqref{eq:epsprime}. 

Similarly, by exploiting \eqref{eq:matbound} we have
\begin{equation}
|\varepsilon_k''|
\leq
\frac{4\,c_\lambda\,\lambda}{1-\lambda}\sum_{j=1}^K v_j\,\gamma_j^2\,\Delta_j+
2\,c_\lambda^2 \sum_{\tau=1}^\infty \lambda^{2\tau} \sum_{j=1}^K \gamma_j^2\,\Delta_j,  
\label{eq:2bound}
\end{equation}
which guarantees the convergence of the series \eqref{eq:epsprimeprime}. 
\end{IEEEproof}

\section{Proof of Theorem~\ref{th:HyperCos}}
\label{app:1}   

Developing the recursion in \eqref{eq:NB2intSLStep0}-\eqref{eq:NB2socialLearningStep0}, we can write
\begin{align}
\log\frac{\bm{\mu}_{k,t}(\theta_1)}{\bm{\mu}_{k,t}(\theta_2)}&=
\sum_{\tau=1}^t\sum_{j=1}^K [A^\tau]_{jk}\,
\gamma_j\,\log\frac{\ell_j(\bm{x}_{j,t-\tau+1}|\theta_1)}{\ell_j(\bm{x}_{j,t-\tau+1}|\theta_2)}
+
\underbrace{\sum_{j=1}^K [A^t]_{jk} 
\log\frac{\mu_{j,0}(\theta_1)}{\mu_{j,0}(\theta_2)}}_{\triangleq \xi_{k,t}}.
\label{eq:fcnet}
\end{align}
Let us consider first the case that the true hypothesis is $\theta_1$. 
It is readily verified that, in this case, for all $t$ we have
\begin{equation}
\log\frac{\ell_k(\bm{x}_{k,t}|\theta_1)}{\ell_k(\bm{x}_{k,t}|\theta_2)}\sim \mathcal{G}\left(
\Delta_k, 2\Delta_k
\right),   
\end{equation}
where $\mathcal{G}(a,b)$ denotes a Gaussian random variable with mean $a$ and variance $b$, and where $\Delta_k$ is the KL divergence defined by \eqref{eq:KLdeltadef}.
Then, exploiting \eqref{eq:fcnet}, the independence across space and time, and the determinism of the initial beliefs, we get
\begin{equation}
\log\frac{\bm{\mu}_{k,t}(\theta_1)}{\bm{\mu}_{k,t}(\theta_2)}
\sim\mathcal{G}\big(\nu_{k,t}+\xi_{k,t},\, \sigma^2_{k,t}\big),
\end{equation}
where we defined
\begin{align}
\nu_{k,t}&\triangleq
\sum_{\tau=1}^t\sum_{j=1}^K [A^\tau]_{jk}\,\gamma_j\,\Delta_j,
\label{eq:1momdef}\\
\sigma^2_{k,t}&\triangleq
2 \sum_{\tau=1}^t\sum_{j=1}^K ([A^\tau]_{jk})^2\, \gamma^2_j\,\Delta_j.
\label{eq:2momdef}
\end{align}
Recalling that we are assuming that $\theta_1$ is true, the conditional error probability in \eqref{eq:conderr} can be written as
\begin{equation}
p_{k,t}(\theta_1)=\mathbb{P}_{\theta_1}\left[
\log\frac{\bm{\mu}_{k,t}(\theta_1)}{\bm{\mu}_{k,t}(\theta_2)}\leq 0
\right]=
Q\left(
\frac{\nu_{k,t}+\xi_{k,t}}{\sigma_{k,t}}
\right).
\label{eq:pkt1}
\end{equation}
Let us switch to the case where $\theta_2$ is true. To compute the corresponding conditional error probability $p_{k,t}(\theta_2)$, it suffices to exchange $\theta_1$ and $\theta_2$ in the previous calculation. With this flipping, the KL divergences $\Delta_j$ remain unchanged, while the term $\xi_{k,t}$ must be replaced by $-\xi_{k,t}$, yielding
\begin{equation}
p_{k,t}(\theta_2)=\mathbb{P}_{\theta_2}\left[
\log\frac{\bm{\mu}_{k,t}(\theta_2)}{\bm{\mu}_{k,t}(\theta_1)}\leq 0
\right]=
Q\left(
\frac{\nu_{k,t} -\xi_{k,t}}{\sigma_{k,t}}
\right).
\label{eq:pkt2}
\end{equation}
Combining \eqref{eq:pkt1} and \eqref{eq:pkt2}, the total error probability is
\begin{equation}
p_{k,t}=\pi(\theta_1)\,
Q\left(
\frac{\nu_{k,t}+\xi_{k,t}}{\sigma_{k,t}}
\right)
+\pi(\theta_2)\,
Q\left(
\frac{\nu_{k,t}-\xi_{k,t} }{\sigma_{k,t}}
\right).
\label{eq:toterrdec}
\end{equation}
Letting $\varepsilon_{jk,t}\triangleq [A^t]_{jk} - v_j$, using \eqref{eq:1momdef} and \eqref{eq:2momdef}, and introducing the definition
\begin{equation}
\Delta \triangleq\sum_{j=1}^K \Delta_j,
\label{eq:bardel}
\end{equation} 
we can write
\begin{equation}
\nu_{k,t}=
t \; \sum_{j=1}^K v_j\,\gamma_j\,\Delta_j
\;+\;
\underbrace{\sum_{\tau=1}^t\sum_{j=1}^K \varepsilon_{jk,\tau}\,\gamma_j\,\Delta_j}_{\triangleq \varepsilon_{k,t}'}
\label{eq:nukt1}
\end{equation}
and
\begin{align}
\sigma^2_{k,t}&=
2 \sum_{\tau=1}^t\sum_{j=1}^K \left(v_j+\varepsilon_{jk,\tau}\right)^2 \,\gamma_j^2\,\Delta_j
\nonumber\\
&
=
2\,t\,\sum_{j=1}^K v_j^2\,\gamma_j^2\,\Delta_j
+
\underbrace{
4 \sum_{\tau=1}^t\sum_{j=1}^K \varepsilon_{jk,\tau} \,v_j\,\gamma_j^2\,\Delta_j
+
2 \sum_{\tau=1}^t\sum_{j=1}^K \varepsilon_{jk,\tau}^2\,\gamma_j^2\, \Delta_j
}_{\triangleq\varepsilon_{k,t}''}.
\label{eq:sig2kt1}
\end{align}
Recalling that we need to examine the traditional social learning strategy with doubly stochastic combination matrix and the NB$^2$ strategy (with left-stochastic matrix), it is convenient to write \eqref{eq:nukt1} and \eqref{eq:sig2kt1} in a compact form holding for both strategies. To this end, we introduce the quantity
\begin{equation}
\alpha \triangleq v_j\,\gamma_j=
\begin{cases}
1  &\textnormal{for NB$^2$,}
\\
\displaystyle{\frac{1}{K}} &\textnormal{for trad. SL with doubly stoch. $A$.}
\end{cases}
\label{eq:alphadefin}
\end{equation}
Using \eqref{eq:bardel} and \eqref{eq:alphadefin} into \eqref{eq:nukt1} and \eqref{eq:sig2kt1} we obtain, respectively, 
\begin{equation}
\nu_{k,t}=
t \; \alpha\,\Delta
\;+\;
\varepsilon_{k,t}',\qquad
\sigma^2_{k,t}=
2\,t\,\alpha^2 \Delta
+
\varepsilon_{k,t}'',
\end{equation}
which, substituted into \eqref{eq:toterrdec}, yields
\begin{align}
p_{k,t}
&=
\pi(\theta_1) 
Q\left(
\frac{t\,\alpha\,\Delta +\varepsilon_{k,t}'+\xi_{k,t}}{\sqrt{2\,t\,\alpha^2\,\Delta+\varepsilon_{k,t}''}}
\right)
+\pi(\theta_2)
Q\left(
\frac{t\,\alpha\,\Delta + \varepsilon_{k,t}'-\xi_{k,t} }{\sqrt{2\,t\,\alpha^2\,\Delta+\varepsilon_{k,t}''}}
\right).
\label{eq:pktprimasenzalabel}
\end{align}
The derivation for the centralized MAP rule is similar. We start from the true posterior in \eqref{eq:truepost} to obtain the identity
\begin{equation}
\log\frac{\bm{\mu}^\star_{t}(\theta_1)}{\bm{\mu}^\star_{t}(\theta_2)}=
\underbrace{\log\frac{\pi(\theta_1)}{\pi(\theta_2)}}_{\triangleq \xi}
+\sum_{\tau=1}^t\sum_{k=1}^K \log\frac{\ell_k(\bm{x}_{k,\tau}|\theta_1)}{\ell_k(\bm{x}_{k,\tau}|\theta_2)},
\label{eq:optBayes}
\end{equation}
and reasoning as before it is straightforward to compute the
total error probability for the optimal system:
\begin{equation}
p^\star_{t}=\pi(\theta_1)\,
Q\left(
\frac{t\Delta +\xi }{\sqrt{2 t\Delta }}
\right)+\pi(\theta_2)\,
Q\left(
\frac{t\Delta -\xi}{\sqrt{2 t\Delta }}
\right).
\label{eq:toterropt}
\end{equation}
To prove the claim of the theorem, we now evaluate the ratio between the  error probability $p_{k,t}$ from \eqref{eq:toterrdec} and the optimal error probability $p^\star_t$ from \eqref{eq:toterropt}. 
It is convenient to introduce the following auxiliary variables:
\begin{equation}
y_{1,t}\triangleq \frac{t\Delta + ( \varepsilon_{k,t}' + \xi_{k,t})/\alpha}
{\sqrt{2\,t\,\Delta +\varepsilon_{k,t}''/\alpha^2}},
\;\;\;
y_{2,t}\triangleq \frac{t\Delta + (\varepsilon_{k,t}'-\xi_{k,t})/\alpha}
{\sqrt{2\,t\,\Delta +\varepsilon_{k,t}''/\alpha^2}}
\label{eq:yvardef}
\end{equation}
and
\begin{equation}
z_{1,t}\triangleq \frac{t\Delta+\xi}{\sqrt{2 t\Delta }},
\qquad
z_{2,t}\triangleq\frac{t\Delta-\xi}{\sqrt{2 t\Delta }},
\label{eq:zvardef}
\end{equation}
which, substituted into \eqref{eq:pktprimasenzalabel} and \eqref{eq:toterropt}, allow us to write
\begin{equation}
\frac{p_{k,t}}{p^\star_t}=
\frac{\pi(\theta_1)Q(y_{1,t})+\pi(\theta_2)Q(y_{2,t})}
{\pi(\theta_1)Q(z_{1,t})+\pi(\theta_2)Q(z_{2,t})}.
\label{eq:pktptratio}
\end{equation}
We recall the following bounds:
\begin{equation}
\frac{1}{\sqrt{2\pi}}\frac{x}{1+x^2}e^{-\frac{x^2}{2}}
\leq 
Q(x)
\leq
\frac{1}{\sqrt{2\pi}}\frac{1}{x}e^{-\frac{x^2}{2}},\qquad x>0.
\label{eq:Qfunbounds}
\end{equation}
Applying the upper bound to the numerator of \eqref{eq:pktptratio} and the lower bound to the denominator, we obtain
\begin{equation}
\frac{p_{k,t}}{p^\star_t}\leq
\frac{
\displaystyle{
\frac{\pi(\theta_1)}{y_{1,t}}
e^{-\frac{y_{1,t}^2}{2}}
+
\frac{\pi(\theta_2)}{y_{2,t}}e^{-\frac{y_{2,t}^2}{2}}
}
}
{
\displaystyle{
\frac{\pi(\theta_1)\,z_{1,t}}{1+z^2_{1,t}}
e^{-\frac{z_{1,t}^2}{2}}
+
\frac{\pi(\theta_2)\,z_{2,t}}{1+z^2_{2,t}}
e^{-\frac{z_{2,t}^2}{2}}
}
}.
\label{eq:pktpstarationewwithdefs}
\end{equation}
From the explicit expressions for the variables $y_{1,t}$, $y_{2,t}$, $z_{1,t}$, and $z_{2,t}$ in \eqref{eq:yvardef} and \eqref{eq:zvardef}, we obtain the following identities:
\begin{equation}
\frac{y^2_{1,t}}{2}
=
\dfrac
{t\Delta }
{4+\dfrac{2\,\varepsilon_{k,t}''}{\alpha^2\,t\,\Delta}}
+\underbrace{
\dfrac{
\dfrac{(\varepsilon_{k,t}' + \xi_{k,t})^2}{\alpha^2\,t\,\Delta} 
+
2\, \dfrac{\varepsilon_{k,t}' + \xi_{k,t} }{\alpha}
}
{4+\dfrac{2\,\varepsilon_{k,t}''}{\alpha^2\,t\,\Delta}}
}_{
\triangleq f_{1,t}
},
\label{eq:formula1}
\end{equation}
\begin{equation}
\frac{y^2_{2,t}}{2}
=
\frac
{t\Delta }
{4+\dfrac{2\,\varepsilon_{k,t}''}{\alpha^2\,t\,\Delta}}
+
\underbrace{
\frac{
\dfrac{(\varepsilon_{k,t}' - \xi_{k,t})^2}{\alpha^2\,t\,\Delta} 
+
2\, \dfrac{\varepsilon_{k,t}'-\xi_{k,t}}{\alpha}
}
{4+\dfrac{2\,\varepsilon_{k,t}''}{\alpha^2\,t\,\Delta}}
}_{\triangleq f_{2,t}},
\label{eq:formula2}
\end{equation}
\begin{equation}
\frac{z^2_{1,t}}{2}
=
\frac{t\Delta}{4}
+\underbrace{
\frac{\xi^2 }{4\,t\,\Delta}
+\frac{\xi}{2}}_{\triangleq g_{1,t}},
\qquad
\frac{z^2_{2,t}}{2}=
\frac{t\Delta}{4}
+\underbrace{
\frac{\xi^2 }{4\,t\,\Delta}
-\frac{\xi}{2}}_{\triangleq g_{2,t}}.
\label{eq:formula34}
\end{equation}
Using \eqref{eq:formula1}--\eqref{eq:formula34}, we can recast  \eqref{eq:pktpstarationewwithdefs} into the compact form
\begin{align}
\frac{p_{k,t}}{p^\star_t}
&\leq
\exp\left\{
\dfrac{t\Delta}{4}
\left(
\dfrac
{\varepsilon_{k,t}''}
{2\,\alpha^2\,t\,\Delta+\varepsilon_{k,t}''}
\right)
\right\}
\frac
{\pi(\theta_1)
\displaystyle{
\frac{e^{-f_{1,t}}}{y_{1,t}}
}
+
\displaystyle{
\pi(\theta_2)\frac{e^{-f_{2,t}}}{y_{2,t}}
}
}
{
\dfrac{\pi(\theta_1)\,z_{1,t}}{1+z^2_{1,t}}
\,e^{-g_{1,t}}
+
\dfrac{\pi(\theta_2)\,z_{2,t}}{1+z^2_{2,t}}
\,e^{-g_{2,t}}
}.
\label{eq:pktptratioconvenient}
\end{align}
Observe now that from Lemma~\ref{lem:classicmatlem} we have
\begin{equation}
\lim_{t\rightarrow\infty} \xi_{k,t} =\xi_{\mathrm{net}}, 
\label{eq:xixikappatconverg2}
\end{equation}
where $\xi_{k,t}$ and $\xi_{\mathrm{net}}$ are defined in \eqref{eq:fcnet} and \eqref{eq:xixinet}, respectively.
In view of Lemma~\ref{lem:neterr}, we can write
\begin{equation}
\lim_{t\rightarrow\infty} \varepsilon_{k,t}'=\varepsilon_k',
\qquad
\lim_{t\rightarrow\infty} \varepsilon_{k,t}''=\varepsilon_k'',
\label{eq:epsiepsikappatconverg2}
\end{equation}
where $\varepsilon_{k,t}'$, $\varepsilon_{k,t}''$, $\varepsilon_k'$, and $\varepsilon_k''$ are defined in \eqref{eq:nukt1}, \eqref{eq:sig2kt1}, \eqref{eq:epsprime}, and \eqref{eq:epsprimeprime}.
Using \eqref{eq:xixikappatconverg2} and \eqref{eq:epsiepsikappatconverg2}, we now compute some useful limits.
First, from \eqref{eq:yvardef} and \eqref{eq:zvardef} we get
\begin{equation}
\lim_{t\rightarrow\infty}\frac{\sqrt{t}}{y_{n,t}}
=
\lim_{t\rightarrow\infty}\frac{\sqrt{t}\,z_{n,t}}{1+z^2_{n,t}}
=\sqrt{\frac{2}{\Delta }},\qquad n=1,2.
\label{eq:yzlims}
\end{equation} 
Second, from the definition of $f_{1,t}$, $f_{2,t}$ introduced in \eqref{eq:formula1} and \eqref{eq:formula2}, we can write
\begin{equation}
\lim_{t\rightarrow\infty} 
f_{1,t}=\frac{
\varepsilon_k'+\xi_{\mathrm{net}}
}{2\,\alpha},\qquad
\lim_{t\rightarrow\infty} 
f_{2,t}=\frac{\varepsilon_k' - 
\xi_{\mathrm{net}}
}
{2\,\alpha},
\label{eq:flims}
\end{equation}
Third, using the expressions for $g_{1,t}$, and $g_{2,t}$ introduced in \eqref{eq:formula34}, we obtain
\begin{equation}
\lim_{t\rightarrow\infty} 
g_{1,t}=\frac{\xi}{2},
\qquad
\lim_{t\rightarrow\infty} 
g_{2,t}=-\frac{\xi}{2}.
\label{eq:glims}
\end{equation}
Furthermore, observe that
\begin{equation}
\lim_{t\rightarrow\infty}
\frac{t\Delta}{4}
\left(
\frac
{\varepsilon_{k,t}''}
{2\,\alpha^2\,t\,\Delta+\varepsilon_{k,t}''}
\right)
=
\frac{\varepsilon_k''}{8\,\alpha^2}.
\label{eq:mainexplim}
\end{equation}
Substituting \eqref{eq:yzlims}, \eqref{eq:flims}, \eqref{eq:glims}, and \eqref{eq:mainexplim} into \eqref{eq:pktptratioconvenient}, we get
\begin{align} 
&\limsup_{t\rightarrow\infty}
\frac{p_{k,t}}{p^\star_t}
\leq
e^{
-\frac{\varepsilon_k'}{2\,\alpha}
+
\frac{\varepsilon_k''}{8\,\alpha^2}
}
\frac
{\pi(\theta_1)e^{-\frac{\xi_{\mathrm{net}}}{2\,\alpha}}
+
\pi(\theta_2)e^{\frac{\xi_{\mathrm{net}}}{2\,\alpha}}
}
{
\pi(\theta_1)e^{-\frac{\xi}{2}}
+
\pi(\theta_2)e^{\frac{\xi}{2}}
}.
\label{eq:limsupfirstexp}
\end{align}
Exploiting the definitions of $\varepsilon_k'$ and $\varepsilon_k''$ from \eqref{eq:epsprime} and \eqref{eq:epsprimeprime}, respectively, we obtain the identity (recall that $v_j\,\gamma_j=\alpha$)
\begin{equation}
-\frac{\varepsilon_k'}{2\,\alpha}
+
\frac{\varepsilon_k''}{8\,\alpha^2}
=
\varepsilon_k,
\label{eq:epskident}
\end{equation}
where $\varepsilon_k$ is defined in  \eqref{eq:overallneterr}.
Moreover, from the definition of $\xi$ in \eqref{eq:xixinet}, we have
$\pi(\theta_1)=\pi(\theta_2) \, e^{\xi}$, which, used in the denominator appearing on the RHS of \eqref{eq:limsupfirstexp}, yields
\begin{align}
&\frac
{\pi(\theta_1)e^{-\frac{\xi_{\mathrm{net}}}{2\,\alpha}}
+
\pi(\theta_2)e^{\frac{\xi_{\mathrm{net}}}{2\,\alpha}}
}
{
\pi(\theta_1)e^{-\frac{\xi}{2}}
+
\pi(\theta_2)e^{\frac{\xi}{2}}
}
=\frac
{\frac{\pi(\theta_1)}{\pi(\theta_2)}
e^{-\frac{\xi_{\mathrm{net}}}{2\,\alpha}}
+
e^{\frac{\xi_{\mathrm{net}}}{2\,\alpha}}
}
{
e^{\xi}\,e^{-\frac{\xi}{2}}
+
e^{\frac{\xi}{2}}
}
\nonumber\\
&=\frac
{e^{\xi}
e^{-\frac{\xi_{\mathrm{net}}}{2\,\alpha}}
+
e^{\frac{\xi_{\mathrm{net}}}{2\,\alpha}}
}
{
2\,
e^{\frac{\xi}{2}}
}
=\frac
{e^{\frac{\xi}{2}}
e^{-\frac{\xi_{\mathrm{net}}}{2\,\alpha}}
+
e^{-\frac{\xi}{2}}
e^{\frac{\xi_{\mathrm{net}}}{2\,\alpha}}
}
{2}
\nonumber\\
&=
\cosh
\left(
\frac
{\xi-\xi_{\mathrm{net}}/\alpha}
{2}
\right).
\label{eq:coshdevelop}
\end{align}
Substituting \eqref{eq:epskident} and \eqref{eq:coshdevelop} into \eqref{eq:limsupfirstexp}, we obtain
\begin{align} 
\limsup_{t\rightarrow\infty}
\frac{p_{k,t}}{p^\star_t}
&\leq
\exp\left(
\varepsilon_k
\right)\,
\cosh
\left(
\frac
{\xi-\xi_{\mathrm{net}}/\alpha}
{2}
\right).
\label{eq:limsupfinal}
\end{align}
The analysis to compute \eqref{eq:limsupfinal}
can be repeated to compute a lower bound on the limit inferior. Specifically, by applying the lower (resp. upper) bound in \eqref{eq:Qfunbounds} to the numerator (resp. denominator) of \eqref{eq:pktptratio}, we can write
\begin{equation}
\frac{p_{k,t}}{p^\star_t}\geq
\frac
{
\displaystyle{
\frac{\pi(\theta_1)\,y_{1,t}}{1+y^2_{1,t}}e^{-\frac{y_{1,t}^2}{2}}
+
\frac{\pi(\theta_2)\,y_{2,t}}{1+y^2_{2,t}}e^{-\frac{y_{2,t}^2}{2}}
}
}
{
\displaystyle{
\frac{\pi(\theta_1)}{z_{1,t}}e^{-\frac{z_{1,t}^2}{2}}
+
\frac{\pi(\theta_2)}{z_{2,t}}e^{-\frac{z_{2,t}^2}{2}}
}
},
\end{equation}
which, by repeating the same steps that led to \eqref{eq:limsupfinal}, yields
\begin{align} 
\liminf_{t\rightarrow\infty}
\frac{p_{k,t}}{p^\star_t}
&\geq
\exp\left(
\varepsilon_k
\right)\,
\cosh
\left(
\frac
{\xi-\xi_{\mathrm{net}}/\alpha}
{2}
\right).
\label{eq:liminffinal}
\end{align}
Joining \eqref{eq:limsupfinal} and \eqref{eq:liminffinal}, we have in fact shown the claim of the theorem.

\bibliographystyle{IEEEtran}

\begin{thebibliography}{10}
\providecommand{\url}[1]{#1}
\csname url@samestyle\endcsname
\providecommand{\newblock}{\relax}
\providecommand{\bibinfo}[2]{#2}
\providecommand{\BIBentrySTDinterwordspacing}{\spaceskip=0pt\relax}
\providecommand{\BIBentryALTinterwordstretchfactor}{4}
\providecommand{\BIBentryALTinterwordspacing}{\spaceskip=\fontdimen2\font plus
\BIBentryALTinterwordstretchfactor\fontdimen3\font minus
  \fontdimen4\font\relax}
\providecommand{\BIBforeignlanguage}[2]{{%
\expandafter\ifx\csname l@#1\endcsname\relax
\typeout{** WARNING: IEEEtran.bst: No hyphenation pattern has been}%
\typeout{** loaded for the language `#1'. Using the pattern for}%
\typeout{** the default language instead.}%
\else
\language=\csname l@#1\endcsname
\fi
#2}}
\providecommand{\BIBdecl}{\relax}
\BIBdecl



\bibitem{ourEUSIPCO2025}
F.~Scala, M.~Carpentiero, V.~Matta and A.~H. Sayed, ``On the performance of social learning,'' in \emph{Proc. EUSIPCO}, Palermo, Italy, Sep. 2025, pp. 1040--1044.


\bibitem{mattaSocialLearningOpinion2024}
V.~Matta, V.~Bordignon, and A.~H. Sayed, \emph{Social {{Learning}}: {{Opinion
  Formation}} and {{Decision-Making Over Graphs}}}.\hskip 1em plus 0.5em minus
  0.4em\relax Now Publishers, 2025.
 
\bibitem{bordignonSociallyIntelligentNetworks2024}
V.~Bordignon, V.~Matta, and A.~H. Sayed, ``Socially {{intelligent networks}}:
  {{A}} framework for decision making over graphs,'' \emph{IEEE Sig.
  Process. Mag.}, vol.~41, no.~4, pp. 20--39, Jul. 2024.
 
\bibitem{zhaoLearningSocialNetworks2012}
X.~Zhao and A.~H. Sayed, ``Learning over social networks via diffusion adaptation,'' in \emph{Proc. Asilomar Conference on Signals, Systems, and Computers}, Pacific Grove, CA, USA, 2012, pp. 709--713.
 
\bibitem{jadbabaieNonBayesianSocialLearning2012}
A.~Jadbabaie, P.~Molavi, A.~Sandroni, and A.~{Tahbaz-Salehi},
  ``Non-{{Bayesian}} social learning,'' \emph{Games and Economic Behavior},
  vol.~76, no.~1, pp. 210--225, Sep. 2012.
 
\bibitem{lalithaSocialLearningDistributed2018}
A.~Lalitha, T.~Javidi, and A.~D. Sarwate, ``Social learning and distributed hypothesis testing,'' \emph{IEEE Trans. Inf. Theory}, vol.~64, no.~9, pp. 6161--6179, Sep. 2018.
 
\bibitem{nedicFastConvergenceRates2017}
A.~Nedi{\'c}, A.~Olshevsky, and C.~A. Uribe, ``Fast convergence rates for distributed non-Bayesian learning,'' \emph{IEEE Trans. Automat. Control}, vol.~62, no.~11, pp. 5538--5553, Nov. 2017.
 
\bibitem{Jadbabaie2018}
P.~Molavi, A.~Tahbaz-Salehi, and A.~Jadbabaie, ``A theory of non-{B}ayesian social learning,'' \emph{Econometrica}, vol.~86, no.~2, pp. 445--490, Mar. 2018.

\bibitem{Djuric2015}
Y.~Wang and P.~M. Djurić, ``Social learning with Bayesian agents and random decision making'' in \emph{IEEE Trans. Signal Process.}, vol. 63, no. 12, pp. 3241-3250, Jun. 2015.

\bibitem{krishnamurthySocialLearningBayesian2013}
V.~Krishnamurthy and H.~V. Poor, ``Social learning and Bayesian games in multiagent signal processing: How do local and global decision makers interact?'' \emph{IEEE Signal Process. Mag.}, vol.~30, no.~3, pp. 43--57, May 2013.
 
\bibitem{chamleyRationalHerdsEconomic2004}
C.~Chamley, \emph{Rational {{Herds}}: {{Economic Models}} of {{Social
  Learning}}}.\hskip 1em plus 0.5em minus 0.4em\relax Cambridge University
  Press, 2004.
   
\bibitem{mitraNewApproachDistributed2021}
A.~Mitra, J.~A. Richards, and S.~Sundaram, ``A {{new approach}} to
  {{distributed hypothesis testing}} and {{non-Bayesian learning}}: {{Improved
  learning rate}} and {{Byzantine resilience}},'' \emph{IEEE Trans. Automat. Control}, vol.~66, no.~9, pp. 4084--4100, Sep. 2021.
 
\bibitem{MertAAGA}
M.~Kayaalp, Y.~İnan, E.~Telatar, and A.~H. Sayed, ``On the arithmetic and geometric fusion of beliefs for distributed inference," \emph{IEEE Trans. Automat. Control}, vol.~69, no.~4, pp. 2265--2280, Apr. 2024.

\bibitem{HuangWang}
B.~Huang, I-H.~Wang ``On the price of decentralization in decentralized detection,'' \emph{IEEE Trans. Inf. Theory}, vol.~71, no.~4, pp. 2341--2359, Apr. 2025.
 
\bibitem{bordignonSocialLearningNonBayesian2023}
V.~Bordignon, M.~Kayaalp, V.~Matta, and A.~H. Sayed, ``Social {{learning}} with
  {{non-Bayesian local updates}},'' in \emph{Proc. EUSIPCO}, Helsinki, Finland, Sep. 2023, pp. 1--5.
 
\bibitem{hornJohnson}
R.~A. Horn and C.~R. Johnson, \emph{Matrix Analysis}. Cambridge {U}niversity {P}ress, 2012.
 
\bibitem{coverthomas}
T.~M. Cover and J.~A. Thomas, \emph{Elements of Information Theory}. \hskip 1em plus 0.5em minus
  0.4em\relax John Wiley \& Sons, 1991.
 
\bibitem{KLDPI}
Y.~Polyanskiy and Y.~Wu, \emph{Information Theory: From Coding to Learning}. \hskip 1em plus 0.5em minus 0.4em\relax Cambridge University Press, 2023.
 
\bibitem{demboLargeDeviationsTechniques2009}
A.~Dembo and O.~Zeitouni, \emph{Large Deviations Techniques and Applications}.\hskip 1em plus 0.5em minus 0.4em\relax Springer, 2009.
 
\bibitem{denhollanderLargeDeviations2000}
F.~{den Hollander}, \emph{Large {{Deviations}}}.\hskip 1em plus 0.5em minus
  0.4em\relax AMS, 2000.

\bibitem{Kar2018}
D.~Li, S.~Kar, and S.~Cui, ``Distributed quickest detection in sensor networks via two-layer large deviation analysis,'' \emph{IEEE Internet Things J.}, vol. 5, no. 2, pp. 930--942, Apr. 2018.

\bibitem{SungTongPoor}
Y.~Sung, L.~Tong, and H.~V. Poor, ``Neyman-Pearson detection of Gauss-Markov signals in noise: closed-form error exponent and properties,'' in \emph{IEEE Trans. Inf. Theory}, vol. 52, no. 4, pp. 1354-1365, Apr. 2006.

\bibitem{MouraLDnonGauss}
D.~Bajović, D.~Jakovetic, J.~M.~F. Moura, J.~Xavier, and B.~Sinopoli, ``Large deviations performance of consensus+innovations distributed detection with non-{G}aussian observations,'' \emph{{IEEE} Trans. Signal Process.}, vol.~60, no.~11, pp. 5987--6002, Nov. 2012.

\bibitem{MouraDIperformance}
D.~Bajović, J.~M.~F. Moura, J.~Xavier and B.~Sinopoli, ``Distributed Inference Over Directed Networks: Performance Limits and Optimal Design,'' \emph{{IEEE} Trans. Signal Process.}, vol.~64, no.~13, pp. 3308--3323, Jul. 2016.
 
\bibitem{Consensus}
L.~Xiao and S.~Boyd, ``Fast linear iterations for distributed averaging,'' \emph{Systems \& Control Letters}, vol.~53, no.~1, pp. 65--78, Feb. 2004.
 
\bibitem{BoydFastMixing}
S.~Boyd, P.~Diaconis, P.~Parrillo, and L.~Xiao, ``Fastest mixing Markov
chain on graphs with symmetries,'' \emph{SIAM Journal on Optimization}, vol.~20, no.~2, pp. 792--819, Mar. 2009.
 
\bibitem{sayednewbooks}
A.~H. Sayed, \emph{Inference and Learning from Data}, Cambridge University Press, 2022.
 
\bibitem{MattaTSIPNexactasy}
V.~Matta, P.~Braca, S.~Marano and A.~H. Sayed, ``Distributed detection over adaptive networks: Refined asymptotics and the role of connectivity,'' \emph{IEEE Trans. Signal Inf. Process. Netw.}, vol.~2, no.~4, pp. 442--460, Dec. 2016.

\end{thebibliography}


\end{document}